\documentclass[11pt]{article}
\mathchardef\ordinarycolon\mathcode`\:
\mathcode`\:=\string"8000
\begingroup \catcode`\:=\active
  \gdef:{\mathrel{\mathop\ordinarycolon}}
\endgroup

\usepackage[margin=1in]{geometry}

\usepackage[backend=biber,maxnames=5]{biblatex}
\usepackage[T1]{fontenc}

\usepackage[utf8]{inputenc}
\usepackage{amsmath,amsthm,amssymb}
\allowdisplaybreaks
\usepackage[hyperref]{xcolor}
\definecolor{darkblue}{rgb}{0,0,0.6}
\usepackage[colorlinks=true,linkcolor=darkblue,citecolor=darkblue,urlcolor=darkblue]{hyperref}

\usepackage{thmtools}
\declaretheoremstyle[
spaceabove=6pt, spacebelow=6pt,
headfont=\normalfont\bfseries,
notefont=\mdseries, notebraces={(}{)},
bodyfont=\normalfont,
postheadspace=1em,
qed={}
]{theorem}
\declaretheoremstyle[
spaceabove=6pt, spacebelow=6pt,
headfont=\normalfont\bfseries,
notefont=\mdseries, notebraces={(}{)},
bodyfont=\normalfont,
postheadspace=1em,
qed={\qedsymbol}
]{example}
\declaretheorem[style=theorem]{theorem}
\declaretheorem[style=theorem]{lemma}

\declaretheorem[style=theorem,numbered=no]{notation}
\declaretheorem[style=theorem]{definition}

\declaretheorem[style=example,numbered=no]{Example}
\declaretheorem[style=example,numbered=no]{Remark}
\usepackage{lmodern}
\usepackage{fix-cm}
\usepackage{algorithm}
\usepackage{algpseudocode}
\usepackage{caption}
\newcommand{\Z}{\mathbb{Z}_N^*}

\newcommand{\g}{{\!g}}
\newcommand{\N}{{\!N}}
\newcommand{\A}{{\!A}}
\newcommand{\B}{{\!B}}

\newcommand{\lA}{\lambda_\A}
\newcommand{\lB}{\lambda_\B}
\newcommand{\LA}{L_\A}

\DeclareMathOperator{\rhoOP}{rho}

\DeclareMathOperator{\ord}{ord}
\DeclareMathOperator{\lcm}{lcm}

\newcommand{\qr}[1]{|#1\rangle}

\usepackage{graphicx,subfigure}
\usepackage{cancel}

\usepackage{float}

\newcommand{\qofa}{quantum order-finding algorithm}
\newcommand{\qpfa}{quantum period-finding algorithm}

\newcommand{\sa}{Shor's algorithm}

\author{Daniel Chicayban Bastos\\
Universidade Federal Fluminense\\
Luis Antonio Kowada\footnote{Authors in alphabetical order. See {\tt https://goo.gl/rzBAq9}.}\\
Universidade Federal Fluminense}
\date{November 2nd 2020}

\title{A quantum version of Pollard's Rho\\
of which Shor's Algorithm is a particular case}

\begin{document}
\fontfamily{cmr}\selectfont
\maketitle
\begin{abstract}
Pollard's Rho is a method for solving the integer factorization
problem.  The strategy searches for a suitable pair of elements
belonging to a sequence of natural numbers that yields a nontrivial
factor given suitable conditions.  In translating the algorithm to a
quantum model of computation, we found its running time reduces to
polynomial-time using a certain set of functions for generating the
sequence.  We also arrived at a new result that characterizes the
availability of nontrivial factors in the sequence.  The result has
led us to the realization that Pollard's Rho is a generalization of
Shor's algorithm, a fact easily seen in the light of the new result.
\end{abstract}

%% \setlength{\parskip}{7pt}
%% \setlength{\parindent}{0pt}
%% \newcommand\blfootnote[1]{%
%%   \begingroup
%%   \renewcommand\thefootnote{}\footnote{#1}%
%%   \addtocounter{footnote}{-1}%
%%   \endgroup
%% }
%\blfootnote{Authors in alphabetical order. See {\tt https://goo.gl/rzBAq9}.}

\section{Introduction}\label{sec:intro}

The inception of public-key cryptography based on the factoring
problem~\cite{rabin1979,rsa} ``sparked tremendous interest in the
problem of factoring large integers''~\cite[section~1.1, page~4]{joy}.
Even though post-quantum cryptography could eventually retire the
problem from its most popular application, its importance will remain
for as long as it is not satisfactorily answered.

While public-key cryptosystems have been devised in the last fifty
years, the problem of factoring ``is centuries
old''~\cite{pomerance1982analysis}.  In the nineteenth century, it was
vigorously put that a fast solution to the problem was
required~\cite[section~VI, article~329, page~396]{gauss}.
\begin{quote}\small
All methods that have been proposed thus far are either restricted to
very special cases or are so laborious and prolix that even for
numbers that do not exceed the limits of tables construed by estimable
men [...] they try the patience of even the practiced calculator.  And
these methods can hardly be used for larger numbers.
\end{quote}
Since then, various methods of factoring have been devised, but none
polynomially bounded in a classical model of computation.  Then in
1994 a polynomial-time procedure was given~\cite{shor94,shor97,shor99}
with the {\em catch}~\cite[page~65]{aaronson2008limits} that it needed
a quantum model of computation.  The algorithm was considered ``a
powerful indication that quantum computers are more powerful than
Turing machines, even probabilistic Turing
machines''~\cite[section~1.1.1, page~7]{qc}: the problem is believed
to be hard~\cite[section~1]{lenstra1994}.

%% (*) Humble version
%% Nevertheless, with each new observation on any of the algorithms for
%% solving it, a new light, dim as it may be, is shed on the problem and
%% its implications.  In the next sections, a new perspective over Shor's
%% algorithm is presented.  It is an observation that, however small, has
%% helped us to better understand it.

%% (*) Bold
Nevertheless, with each new observation, a new light is shed on the
problem and its implications.  In the next sections, a new perspective
over Shor's algorithm is presented.  It is an observation that has
helped us to better understand it.
%% --8<---------------cut here---------------end--------------->8---

\subsection*{Contents of this paper}

\begin{itemize}
\item Theorem~\ref{thm:main} in Section~\ref{sec:main} is a new
  result.  The theorem characterizes nontrivial collisions in the
  cycles of sequences produced by the polynomials used in Pollard's
  Rho.  It allows us to write a certain quantum version of the
  strategy, which is presented in Section~\ref{sec:quantum-rho}.  For
  clarity, we provide in Section~\ref{sec:circuit} a description of a
  quantum circuit for the algorithm.

\item The quantum version of Pollard's Rho presented happens to be a
  generalization of \sa\ and this fact is described in
  Section~\ref{sec:restrict}.  This is a new way of looking at \sa,
  but Section~\ref{sec:restrict} does not bring any new result.

\item Readers familiar with Pollard's Rho may skip
  Sections~\ref{sec:pollard-strategy}--\ref{sec:floyd}.  Similarly,
  Sections~\ref{sec:shor}--\ref{sec:extended} only describe Shor's
  algorithm, its original version and its extension to odd orders.

\item Section~\ref{sec:state-of-the-art} is a quick description of how
  to factor an integer with an emphasis on number-theoretical results
  of relevance to \sa.  It serves as a brief summary of scattered
  results in the literature.
\end{itemize}

\section{The state-of-the-art in factoring}\label{sec:state-of-the-art}

Assuming we know nothing about the integer, a reasonable general
recipe to factor an integer $N$ on the classical model of computation
is to try to apply special-purpose algorithms first.  They will
generally be more efficient if $N$ happens to have certain properties
of which we can take advantage.  Special-purpose algorithms include
Pollard's Rho, Pollard's $p - 1$, the elliptic curve algorithm and the
special number field sieve.

As an example of a general strategy, we can consider the following
sequence of methods.  Apply trial division first, testing for small
prime divisors up to a bound $b_1$.  If no factors are found, then
apply Pollard's Rho hoping to find a prime factor smaller than some
bound $b_2 > b_1$.  If not found, try the elliptic curve method hoping
to find a prime factor smaller than some bound $b_3 > b_2$.  If still
unsuccessful, then apply a general-purpose algorithm such as the
quadratic sieve or the general number field sieve~\cite[chapter~3,
  section~3.2, page~90]{hac}.

Since 1994, due to the publication of Shor's algorithm, a quantum
model of computation has been an important part of the art of
factoring because Shor's algorithm gives us hope of factoring integers
in polynomial-time~\cite{shor94,shor97,shor99}: to make it a reality,
we need to build large quantum computers.

Shor's algorithm is a probabilistic algorithm --- it succeeds with a
probability~\cite[section~5, page~317]{shor99} of at least $1 -
1/2^{k-1}$, where $k$ is the number of distinct odd prime factors of
$N$.  The integer factoring problem reduces~\cite{miller1976}, via a
polynomial-time transformation, to the problem of finding the order of
an element $x$ in the multiplicative group $\Z$.  The order $r =
\ord(x, N)$ of an element $x \in \Z$ is the smallest positive integer
$r$ such that $x^r = 1 \bmod{N}$.

If $r = \ord(x, N)$, for some integer $x \in \Z$, \sa\ finds a
factor by computing the greatest common divisor of $x^{r/2} -1$ and
$N$, that is, it computes $\gcd(x^{r/2} - 1, N)$, implying $r$ must be
even.  The essence of the strategy comes from the fact that
\begin{equation}\label{eq:shor}
 (x^{r/2} -1) (x^{r/2} + 1) = x^{r} - 1 = 0 \bmod{N}.
\end{equation}

It is easy to see from Equation~(\ref{eq:shor}) that if $x^{r/2}
\equiv -1 \bmod{N}$ then the equation is trivially true, leading the
computation of the $\gcd$ to reveal the undesirable trivial factor
$N$.  So Shor's algorithm needs not only an even order, but also
$x^{r/2} \not\equiv -1 \bmod{N}$.

If $\ord(x, N)$ happens to be odd, the algorithm must try a different
$x$ in the hope that its order is even.  To that end, an improvement
has been proposed~\cite{leander2002} to the effect that choosing $x$
such that $J(x, N) = -1$, where $J(x, N)$ is the Jacobi symbol of $x$
over $N$, lifts the lower bound of the probability of success of the
algorithm from $1/2$ to $3/4$.  This improvement tries to steer clear
from odd orders, but other contributions~\cite{lawson2015} show odd
orders can be used to one's advantage.  For example, if $r$ is odd but
$x$ is a square modulo $N$, then finding $y$ such that $y = x^2
\bmod{N}$ can lead to extracting a factor from $N$ by computing
$\gcd(y^r -1, N)$ and $\gcd(y^r + 1, N)$ if $N$ is of the form $N =
p_1 p_2$, where $p_1, p_2 \equiv 3 \bmod{4}$.  Compared to choosing
$J(x,N)=-1$, the improvement is a factor of $1 - 1/(4\sqrt{N})$, which
is too small: it is less than 1\% when $N$ has seven bits or more.

More emphatically, an extension of \sa\ has been
proposed~\cite{johnston2017} that uses any order of $x$ modulo $N$
satisfying $\gcd(x, N) = 1$ as long as a prime divisor of the order
can be found.  The work also includes sufficient
conditions~\cite[section~4, pages 3--4]{johnston2017} for when a
successful splitting of $N$ should occur.  Let us state the result.
Let $N = AB$, where $A, B$ are coprime nontrivial factors of $N$.  Let
$r = \ord(x, N)$ and $d$ be a prime divisor of $r$.  Suppose
  \[x^{r/d} \equiv 1 \bmod{A} \text{\qquad and\qquad}  x^{r/d} \not\equiv 1 \bmod{B}.\]
Then $1 < \gcd(x^{r/d} - 1, N) < N$.  A different perspective of this
result has been given~\cite[section~3]{extension} and the equivalence
of both perspectives has been established~\cite[section~3]{extension}.
Moreover, a generalization of Equation~(\ref{eq:shor}) that naturally
leads us to the extension~\cite{johnston2017} has been
provided~\cite[section~4]{extension}, making it immediately clear why
the extension of the algorithm works.  For example, it has been
observed~\cite[section~3, page~3]{johnston2017} that when $2$ is a
divisor of $r$, then $(x^{r/2} - 1)$ and $(x^{r/2} + 1)$ are the
nontrivial factors of $N=pq$, but if $3$ is also a divisor of $r$,
then even assuming that $(x^{r/3} - 1)$ is a nontrivial factor, the
other nontrivial factor is not $(x^{r/3} + 1)$.  To find the other
factor we need to use the general form~\cite[section~4]{extension} of
Equation~(\ref{eq:shor}), which is
 \[(x^{r/d} -1) \sum_{i=0}^{d - 1} x^{ir/d} = x^{r} - 1 \equiv 0 \bmod{N}.\]
In particular, if $3$ divides $r$, the other factor is $1 + x^{r/3} +
x^{2r/3}$, since $x^r - 1 = (x^{r/3} - 1)(1 + x^{r/3} + x^{2r/3})$.

The \qofa\ is also a probabilistic procedure.  There are times when
the answer given by the procedure is a divisor of the order, not the
order itself, that is, the procedure sometimes fails.  It has been
shown~\cite{xu2018} how to use these failed runs of the \qofa\ to
split $N$.  Suppose, for instance, that the \qofa\ produces a divisor
$c$ of $\ord(x, N)$.  Then $\gcd(x^c - 1, N)$ might yield a nontrivial
factor.  In this respect, we present in Section~\ref{sec:main} a new
theorem that shows when $x^c - 1$ shares a nontrivial factor with $N$.

Particular properties of $N$ have also been investigated providing
special-purpose variations of \sa.  For instance, if $N=pq$ is a
product of two distinct safe primes greater than 3, then $N$ can be
factored by a variation of \sa\ with probability approximately $1 -
4/N$ using a single successful execution of the \qofa, leading to the
fact that the product of two distinct safe primes is easy to
factor~\cite[section~1, page~2]{grosshans2015}.  In this direction, it
has been shown that if $N$ is a product of two safe primes, not
necessarily distinct, then any $x \bmod{N}$ such that $\gcd(x, N) = 1$
and $1 < x$ allows one to find a nontrivial factor of $N$ with a
single successful execution of the \qofa\ if and only if $J(x, N) =
-1$, where $J(x, N)$ is the Jacobi symbol of $x$ over $N$.

We also observe that, as a description of the state-of-the-start,
Shor's algorithm is not the complete story.  The subject is richer.
Interesting results have been published that are not based on Shor's
algorithm: GEECM, a quantum version of the elliptic curve method using
an Edwards curve, ``is often much faster than Shor's algorithm and all
pre-quantum factorization algorithms''~\cite[section 1, page
  2]{pqrsa}. Also, Shor's method ``is not competitive with [other
  methods that excel] at finding small primes''~\cite[section 2, page
  6]{pqrsa}.  The state-of-the-art in factoring is not only concerned
with large integers, although large integers are obviously of great
importance, given their wide use in cryptography.

\section{The strategy in Pollard's Rho}\label{sec:pollard-strategy}

Pollard's Rho is an algorithm suitable for finding small prime factors
in a composite number $N$ that is not a prime power.  Before applying
the strategy, it should be checked that the number to be factored is
not a prime power, a verification that can be done in
polynomial-time~\cite{aks2004}\cite[chapter~3, note~3.6,
  page~89]{hac}.  Throughout this paper, we assume these verifications
are performed before Pollard's Rho is applied.

Let us begin with an important well-known fact used in Pollard's Rho.

\begin{theorem} \label{thm:col-mod-p}  
Let $N = AB$, where $A, B$ are coprime nontrivial factors of $N$.  Let
$(N_k)$ be an infinite sequence of natural numbers reduced modulo $N$.
Let $(A_k)$ be the sequence of integers obtained by reducing modulo
$A$ each element $n_k \in (N_k)$.  If $a_i = a_j\in (A_k)$ then
  \[1 < \gcd(n_i - n_j, N) \le N,\]
for $n_i, n_j \in (N_k)$, where $\gcd$ represents the greatest common
divisor among its arguments.
\end{theorem}

\begin{proof}
If $a_i\equiv a_j \bmod{A}$ then $A$ divides $a_i - a_j$.  Therefore,
$A$ divides $\gcd(a_i - a_j, N)$ and $1 < A$.
\end{proof}

%% \begin{proof}
%% We begin by dividing $n_i$ and $n_j$ by $A$.  The Division
%% Algorithm~\cite[section~2.2, theorem~2.1, page~17]{burton2010elementary}
%% guarantees that
%% %
%% \begin{align*}
%%         n_i &= q_1 A + a_i = q_1 A + r\\
%%         n_j &= q_2 A + a_j = q_2 A + r,
%% \end{align*}
%% where $q_1, q_2$ are the quotients and $r = a_i = a_j$ the remainders.
%% Subtracting one equation from the other, we get %align by =
%% %v
%% \begin{align*}
%%         n_i - n_j &= (q_1 A + r) -(q_2 A + r)\\
%%                   &= (q_1 - q_2) A,
%% \end{align*}
%% %e
%% implying $A$ is a factor of $n_i - n_j$.  Since $A$ is a nontrivial
%% factor of $N$, then $1 < A$ and, therefore,
%% %
%%   \[1 < \gcd(n_i - n_j, N) \le N,\]
%% %
%% as desired.
%% \end{proof}

In other words, when $a_i = a_j \in (A_k)$, then $n_i - n_j$ shares a
common factor with $N$, where $n_i, n_j \in (N_k)$.  %% The case $i = j$
%% is uninteresting because $n_i = n_j$ and $\gcd(n_i - n_j, N) = \gcd(0,
%% N) = N$.

Given a function $f\colon S \to S$, where $S$ is a finite nonempty
subset of the natural numbers, if the infinite sequence $(N_k)$ is
generated by the rule $n_{i+1} = f(n_i) \bmod{N}$, for all $i \ge 1$,
then $(N_k)$ contains a cycle.  Consequently, $(A_k)$ contains a
cycle, so there are indices $i \neq j$ such that $a_i = a_j \in
(A_k)$.  Moreover, whenever $a_i = a_j$, it follows that $1 < \gcd(n_i
- n_j, N) \le N$.  We say these pairs $(a_i, a_j)$ and $(n_i, n_j)$
are collisions\footnote{This terminology comes from the study of hash
  functions~\cite[chapter~5, page~137]{stinson2018}.  When two
  different elements $x, y$ in the domain of a hash function $h$
  satisfy $h(x) = h(y)$, we say $(x, y)$ is a collision.  We extend
  the terminology by adding the qualifiers ``trivial'' and
  ``nontrivial''.  Choose an element $n_i$ in the cycle of $(N_k)$.
  Checking the next elements $n_{i+1}, n_{i+2}, ...$, one by one, if
  we eventually find that $1 < \gcd(n_i - n_j, N) \le N$ for some $j >
  i$, then the pair $(i, j)$ is called a
  collision.}.

\begin{definition}
\label{def:collision}
Let $(N_k)$ be a sequence of integers reduced modulo $N$.  If
$\gcd(n_i - n_j, N) = N$, where $n_i, n_k \in (N_k)$ for indices $i$
and $j$, we say $(i, j)$ is a {\em trivial collision} relative to
$(N_k)$.  If $1 < \gcd(n_i - n_j, N) < N$, we say $(i, j)$ is a {\em
  nontrivial collision} relative to $(N_k)$.  When context makes it
clear, we refrain from explicitly saying which sequence the collision
refers to.
\end{definition}

As an immediate application of Definition~\ref{def:collision}, we may
define the ``Pollard's Rho Problem'' as the task of finding a
nontrivial collision in an infinite sequence $n_0, n_1, ...$ of
natural numbers reduced modulo~$N$.

As Theorem~\ref{thm:col-mod-p} asserts, a collision in $(A_k)$
provides us with enough information to find a collision in $(N_k)$,
out of which we might find a nontrivial factor.  The smaller the cycle
in $(A_k)$, the faster we would find a collision in $(A_k)$.  Using an
arbitrary function $f\colon S \to S$ to generate $(N_k)$ via a rule
$n_{i+1} = f(n_i) \bmod{N}$, where $S$ is a finite nonempty subset of
the natural numbers, we cannot guarantee that the cycle in $(A_k)$ is
smaller than the cycle in $(N_k)$, but if we take $f$ to be a
polynomial of integer coefficients, such as $f(x) = x^2 + 1$, then the
cycle in $(A_k)$ is smaller than the cycle in~$(N_k)$ with high
probability~\cite[section~3.2.2, note~3.8, page~91]{hac}.  This
importance of a polynomial of integer coefficients for Pollard's Rho
is established by the next theorem~\cite[section~6.6.2,
  pages~213--215]{stinson2018}.

\begin{theorem} \label{thm:poly-importance} 
Let $N$ be a composite number having $p$ as a prime divisor.  Let
$f(x)$ be a polynomial of integer coefficients.  Fix $n_0 < N$ as the
initial element of the infinite sequence $(N_k)$ generated by the rule
$n_{k+1} = f(n_k) \bmod{N}$ for all $k \geq 0$.  If $n_i = n_j
\bmod{p}$, then $n_{i+\delta} = n_{j + \delta} \mod p$ for all indices
$\delta \geq 0$.
\end{theorem}

\begin{proof} 
Suppose $n_i = n_j \bmod{p}$.  Since $f$ is a polynomial of integer
coefficients, then $f(n_i) = f(n_j) \bmod{p}$.  By definition,
$n_{i+1} = f(n_i) \bmod{N}$, so %cm
%v
  \[n_{i+1} \bmod{p} = (f(n_i) \bmod{N}) \bmod{p} = f(n_i) \bmod{p},\]
%e
because $p$ divides $N$. Similarly, $n_{j+1} \bmod{p} = f(n_j)
\bmod{p}$. Hence, $n_{i+1} = n_{j+1} \bmod{p}$.  By repeating these
steps $\delta$ times, we may deduce %cm
%v
  \[n_i = n_j \bmod{p} \implies n_{i+\delta} = n_{j+\delta} \bmod{p}\]
%e
for all $\delta \geq 0$, as desired.
\end{proof}

\begin{figure}[htb]
  \centering
  \begin{minipage}{.55\textwidth}
    \centering \includegraphics[width=\linewidth]{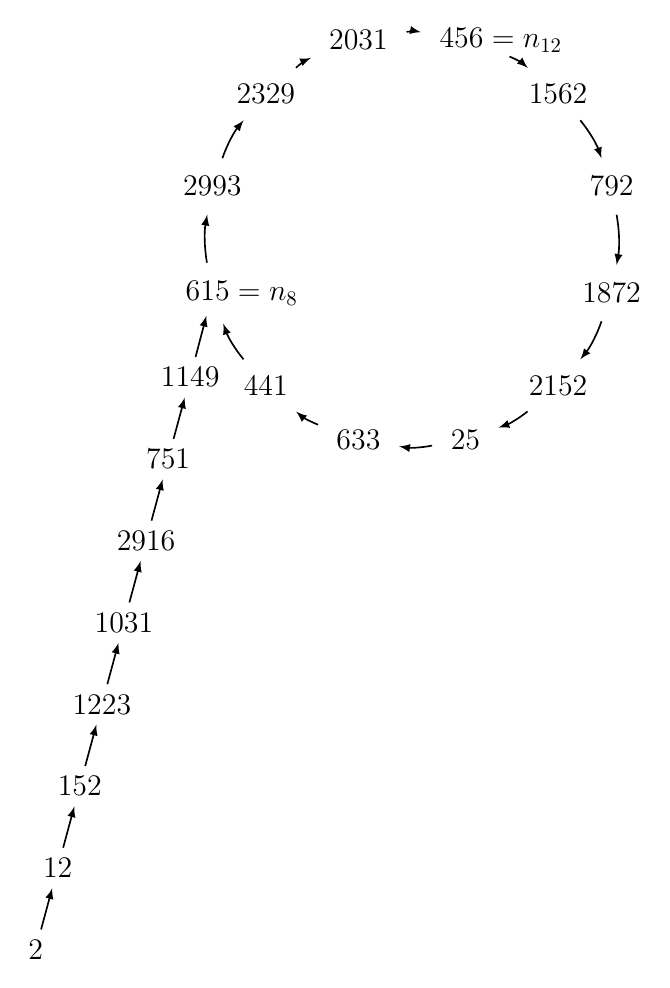}
  \end{minipage}%
  \begin{minipage}{.40\textwidth}
    \centering \includegraphics[width=\linewidth]{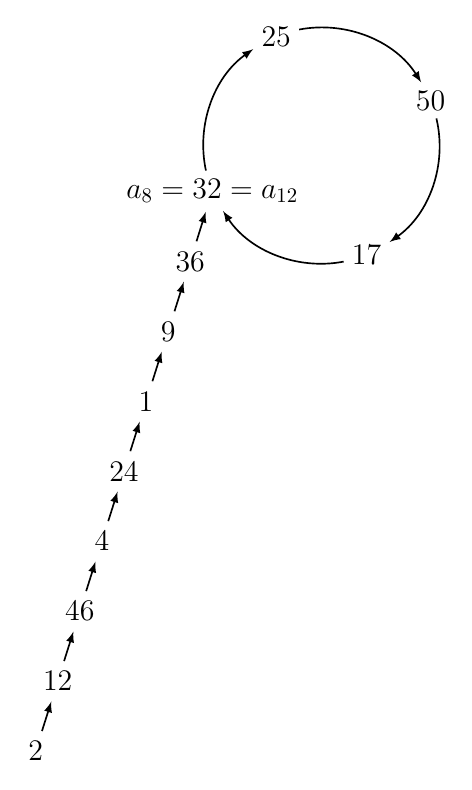}
  \end{minipage}
  \caption{The periodic sequences $n_{k+1} = n_k^2 + 8 \bmod{3127}$ and
    $a_{k+1} = a_k^2 + 8 \bmod{53}$ with $n_0 = a_0 = 2$ and a collision
    modulo $53$ at $(a_8, a_{12})$ related to the collision modulo $N$
    at $(n_8, n_{12})$.}
  \label{fig:3127-53}
\end{figure}

\begin{figure}[htb]
  \centering
  \begin{minipage}[b]{.30\linewidth}
    \centering \includegraphics[width=\textwidth]{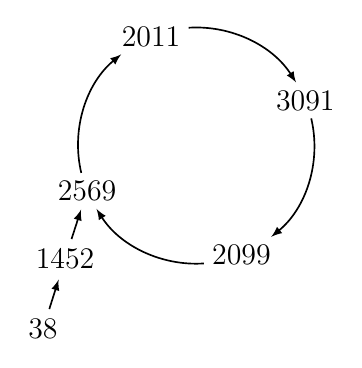}
    \label{fig:random-3551}
  \end{minipage}%
  \begin{minipage}[b]{.30\linewidth}
    \centering \includegraphics[width=\textwidth]{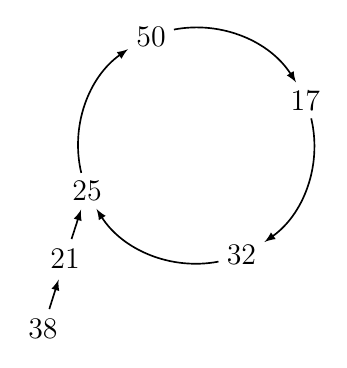}
    \label{fig:random-3551-53}
  \end{minipage}
  \begin{minipage}[b]{.30\linewidth}
    \centering \includegraphics[width=\textwidth]{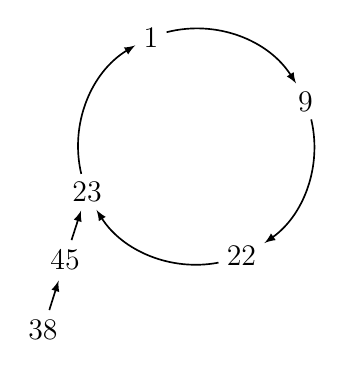}
    \label{fig:random-3551-67}
  \end{minipage}
  \caption{The sequences $n_k^2 + 8 \bmod{3551}$, $a_k^2 + 8
    \bmod{53}$, $b_k^2 + 8 \bmod{67}$ with $n_0 = a_0 = b_0 = 38$, an
    example of terrible luck for Pollard's Rho due to the cycles
    modulo $3551$, $53$ and $67$ of equal length.}
  \label{fig:random-3551-53-67}
\end{figure}

To illustrate the importance of Theorem~\ref{thm:poly-importance}, let
us look at an example.  Using Figure~\ref{fig:3127-53} as a guide, let
$(N_k)$ be the sequence generated by the rule $n_{k+1} = n_k^2 + 8
\bmod{3127}$ with initial value $n_0 = 2$.  If we knew that indices $i
= 8$ and $j = 12$ of the sequence $(A_k)$ provided a collision
relative to $(A_k)$, then we would compute $\gcd(n_8 - n_{12}, 3127) =
\gcd(615 - 456, 3127)$ and find $53$ as a nontrivial factor of $3127$.
But notice the pair $(n_8, n_{12})$ is not the only collision in the
cycle of $(N_k)$.  Indeed, the pairs $(9, 13)$, $(10, 14)$, $(11,
15)$, $(12, 16)$, $(13, 17)$, $(14, 18)$ are also collisions.  In
fact, since $\lA = 4$ is the length of the cycle of $(A_k)$, we get a
collision $(n_i, n_j)$ as long as $j - i$ is a multiple of $\lA$,
which greatly increases the probability we will find a collision in
$(N_k)$ compared to the case in which $(A_k)$ has the same cycle
length as that of $(N_k)$.

In Section~\ref{sec:main}, we present a new result
(Theorem~\ref{thm:main}) that characterizes nontrivial collisions in
terms of the lengths of the cycles of the sequences $(A_k)$ and
$(B_k)$, where $(A_k)$ is the sequence obtained by reducing modulo $A$
each element $n_k \in (N_k)$ and similarly for $(B_k)$.  The result
states that there is a nontrivial collision if and only if $\lA \ne
\lB$, where $\lA$ is the length of the cycle contained in the sequence
$(A_k)$ and similarly for $\lB$.  In particular, if $m$ is a multiple
of $\lA$ but not of $\lB$, then $(i, i + m)$ is a nontrivial
collision whenever $n_i$ is an element of the cycle contained in
$(N_k)$.  Finding a nontrivial collision by effectively getting a hold
of $m$ produces the nontrivial factor $A$ of $N$.  Thus, we may
equivalently understand the Pollard's Rho Problem as the task of
finding a suitable $m$ that is a multiple of the length of the cycle
in $(A_k)$ but not a multiple of the length of the cycle in $(B_k)$.

%% An efficient method for finding the length of a cycle or a multiple of
%% the cycle is not currently known for the classical model of
%% computation.  
It is not obvious how to efficiently solve the Pollard's Rho Problem
in the classical model of computation.  Floyd's algorithm for
cycle-detection~~\cite[chapter 3, exercise~6b, page~7]{knuthv2}
provides us with a set of pairs that are collisions, not all of which
are nontrivial.  Thus, by using Floyd's algorithm, Pollard's Rho is
able to make educated guesses at pairs that might be nontrivial
collisions, optimizing a search that would otherwise be a brute-force
approach.

Typical sequences chosen for Pollard's Rho are generated by polynomial
functions of the form $n_{k+1} = (n_k^e + c) \bmod{N}$. Very little is
known about these polynomials, but it is clear that $c = -2$ should
not be used if $e = 2$. If $c = -2$ and $e = 2$, then $n_{k+1} = 2$
whenever $n_k = 2$, closing a cycle of length $1$.  If $(N_k)$ has a
cycle of length $1$, then $(A_k)$ has a cycle of length $1$ and so
does $(B_k)$, rendering trivial all collisions.
Figure~\ref{fig:random-3551-53-67} illustrates the case when all
cycles have the same length.

Many implementations choose the exponent $e = 2$.  The argument is
that polynomials of greater degree are more expensive to compute and
not much more is known about them than it is about those of second
degree~\cite[section~19.4, page~548]{gathen}.

\section{Floyd's algorithm as a strategy for the Pollard's Rho Problem}\label{sec:floyd}

Let us now discuss the generation of pairs that are collision
candidates.  The objective of Pollard's Rho is to find a nontrivial
collision, so any refinement we can make in the set of all possible
pairs is useful.  The candidates for collision are formed by pairing
elements of a sequence that cycles.  If nontrivial collisions are
nonexistent in the sequence, as illustrated by
Figure~\ref{fig:random-3551-53-67}, we must have a way to give up on
the search, lest we cycle on forever.  How can we avoid cycling on
forever?  A trivial strategy is obtained by storing in memory each
element seen and stopping when the next in the sequence has been seen
before.  In more precision: create a list $L$ and store $x_0$ in $L$,
where $x_0$ is the first element of the sequence.  Now, for $i = 1$,
compute $y = f(x_i)$ and verify whether $y \in L$.  If it is, we found
$f(x_i) = f(x_j)$ where $i \neq j$, hence a collision is $(x_i, x_j)$.
Otherwise, store $y \in L$, set $i = i + 1$ and repeat.  The
verification process of whether $y \in L$ can be done efficiently by
sorting the values $f(x)$ in a list $K$ and applying a binary search
to check whether $y$ was previously seen.  Binary search guarantees no
more than $\Theta(\lg|K|)$ comparisons would be made~\cite[section
  4.2.2, algorithm~4.3, page~125]{stinson2018}, where $|K|$ represents the
cardinality of $K$.  However, in this strategy, the space required for
the list $K$ grows linearly in $N$, an exponential amount of memory.
Floyd's algorithm~\cite[chapter 3, exercise~6b, page~7]{knuthv2}, on
the other hand, requires essentially just two elements of the sequence
to be stored in memory.  Let us see how this is possible.

\begin{algorithm}
\caption{Pollard's Rho using a polynomial $f(x)$ with initial value
  $x_0$.  The strategy for finding collisions is the one provided by
  Floyd's algorithm for cycle-detection.  The variable $a$ represents
  the position of Achilles while $t$ represents the tortoise's.  The
  procedure gives up on the search as soon as it finds a trivial
  collision.  Assume $N$ is not a prime power.}
\label{alg:rho}
\begin{algorithmic}[0]
\Procedure{$\rhoOP$}{$N, f, x_0$}:
  \State $a \gets t \gets x_0$
  \Loop
    \State $t \gets f(t) \bmod{N}$
    \State $a \gets f(f(a)) \bmod{N}$
    \State $d \gets \gcd(t - a, N)$
    \If{$d = N$}
      \State{\Return none}
    \EndIf
    \If{$1 < d < N$}
      \State{\Return $d$}
    \EndIf
  \EndLoop
\EndProcedure
\end{algorithmic}
\end{algorithm}

Picture the sequence of numbers containing a cycle as a race track.
Let us put two old friends, Achilles and the tortoise, to compete in
this race.  While the tortoise is able to take a step at every unit of
time, Achilles is able to take two steps.  By ``step'', we mean a jump
from one number in the sequence to the next.  After the starting gun
has fired, will Achilles ever be behind the tortoise?  He eventually
will because the track is infinite and contains a cycle.

Floyd's strategy proves Achilles catches the tortoise by eventually
landing on the same element of the sequence as the tortoise.
Moreover, they meet at an index that is a multiple of the length of
the cycle.  Let us see why.

Let $\lambda$ be the length of the cycle and $\mu$ the length of the
tail.  If both Achilles and the tortoise are in the cycle, then they
must have passed by at least $\mu$ elements of the sequence.  They can
never meet outside the cycle because, throughout the tail, Achilles is
always ahead of the tortoise.  Let $k$ be the distance between the
beginning of the cycle and the index at which Achilles meets the
tortoise.  Let $\LA$ be the number of laps Achilles has completed
around the cycle when he meets the tortoise.  Similarly, let $L_t$
describe the number of laps the tortoise has given around the cycle up
until it was caught by Achilles.  What is the distance traveled by
each runner?  Achilles has traveled $\mu + \lambda \LA + k$, while the
tortoise has traveled $\mu + \lambda L_t + k$.  Since Achilles travels
with double the speed of the tortoise, we may deduce
  $\mu + \lambda \LA + k = 2(\mu + \lambda L_t + k)$,
implying
\begin{align}
  \mu + k = \lambda (\LA - 2L_t) = \lambda M, \label{eq:where-they-meet}
\end{align}
%e
where $M = \LA - 2L_t$.  Now, notice $\mu + k$ describes the index of
the sequence where Achilles meets the tortoise.  Therefore,
Equation~(\ref{eq:where-they-meet}) tells us they meet at an index of
the sequence that is a multiple of~$\lambda$.

It is not possible for Achilles to jump the tortoise and continue the
chase.  At each iteration, one step in the distance between them is
reduced.  In particular, if Achilles is one step behind the tortoise,
they both meet in the next iteration.  Thus, since the cycle has
length $\lambda$, Achilles must meet the tortoise in at most $\lambda
- 1$ steps.

\begin{theorem}[name=Robert W. Floyd]\label{thm:floyd}
Given a function $f\colon S \to S$, where $S$ is a nonempty finite set, let
$x_0, x_1, \dots$ be an infinite sequence generated by the rule
$x_{i+1} = f(x_i)$, for $i\ge 0$, where $x_0 \in S$ is given. Since
this sequence has a cycle, let $\mu$ be the length of the tail of the
sequence and $\lambda$ the length of the cycle. Then,
\[i\ge \mu \text{ and } i = n\lambda \text{\qquad if and only if\qquad} x_i = x_{2i},\]
for any natural numbers $n \ge 1$ and $i \ge 1$.
\end{theorem}

\begin{proof}
Suppose $i = n\lambda$, where $n \geq 1$, and $i \geq \mu$, that is,
suppose $i$ is an index in the cycle, so that $x_i = x_{i + t \lambda}$
for each natural $t \geq 0$.  In particular, if $t = n$, we have %cm
%v
  \[x_{2i} = x_{i + i} = x_{i + n \lambda} = x_{i + t\lambda} = x_i,\]
%e
as desired.

Suppose $x_i = x_{2i}$.  By way of contradiction, suppose $i < \mu$,
that is, let us see what happens if $i$ is an index on the tail, a
chunk of the sequence where $x_i = x_j$ if and only if $i = j$.  By
hypothesis, $i \ge 1$.  So, on the tail it cannot happen that $i = 2i$,
since that would imply $i = 0$.  Therefore, $x_i \neq x_{2i}$,
violating the initial assumption.  Thus, $i \geq \mu$, in which case
\begin{equation*}
  i = \mu + [(i - \mu) \bmod{\lambda}] \text{\qquad and\qquad} 2i = \mu + [(2i - \mu) \bmod{\lambda}],
\end{equation*}
where the expression $(i - \mu) \bmod{\lambda}$ describes the index
$i$ relative to the beginning of the cycle, not to the beginning of
the sequence.

Since $x_i = x_{2i}$, then
\begin{equation}\label{eq:*}
  i = \mu + \bigl[(i - \mu) \bmod {\lambda}\bigr] = \mu + \bigl[(2i - \mu)\bmod{\lambda}\bigr] = 2i.
\end{equation}
Subtracting $\mu + \bigl[(i - \mu) \bmod{\lambda}\bigr]$ from both
sides of Equation~(\ref{eq:*}), we get %align by =
\begin{align*}
  0 &= \mu + \bigl[(2i - \mu) \bmod{ \lambda}\bigr] - \mu + \bigl[(i - \mu) \bmod{ \lambda}\bigr]\\
  &= (2i - \mu) - (i - \mu) \bmod{ \lambda}\\
  &= i \bmod{ \lambda},
\end{align*}
%e
implying $i = n\lambda$ for all $n \geq 0$, that is, $i$ is a multiple
of $\lambda$, as desired.
\end{proof}

Applying Floyd's theorem to the problem of finding a collision in the
cycle, we can be sure some collision will occur at the pair $(i, 2i)$,
making certain the procedure terminates.  Refinements of Floyd's
method have been published~\cite{brent1980}\cite[chapter 3,
  exercise~7, page~8]{knuthv2}.  The Pollard's Rho algorithm with
Floyd's method is displayed in Algorithm~\ref{alg:rho}.

Floyd's algorithm has the merit that it takes constant space, since it
essentially needs to hold only two numbers in memory, but the worst
case time-complexity of the strategy is $\Theta(\mu + \lambda)$, where
$\lambda$ is the length of the cycle and $\mu$ is the length of the
tail.

%% So the order is linear with respect a $N$, the composite we wish to
%% factor, but Floyd's algorithm has the merit that it takes constant
%% space, since it essentially needs to hold only two numbers in
%% memory.

We now proceed to the new result that allows us to characterize
nontrivial collisions.  The theorem leads us to a quantum version of
Pollard's Rho of which \sa\ is particular case.

\section{A characterization of nontrivial collisions}\label{sec:main}

The uniqueness claim in the fundamental theorem of
arithmetic~\cite[section~II, article~16,
  page~6]{gauss}~\cite[section~1.3, page~3, sections~2.10, 2.11,
  page~21]{hardy1975intro} guarantees that if two natural numbers have
identical prime factorization, then they are the same number.  The
contrapositive of this fact is that if two numbers are not the same,
there must be at least one prime power which appears in one
factorization but not on the other.  That is, two different numbers
can be distinguished by such prime power.  One way to pinpoint this
distinguishing prime power is to develop the following device.
Let \[N = \prod_s s^{e(s, N)},\] where $s$ is a prime number and $e(s,
N)$ is the exponent of $s$ in the prime factorization of $N$.  If $s$
does not divide $N$, we set $e(s, N) = 0$.  This way all natural
numbers are expressed in terms of all prime numbers.  Such device
allows us to distinguish two natural numbers $N \neq M$ by writing
  \[N = \prod_s s^{e(s, N)} \qquad\text{and}\qquad  M = \prod_s s^{e(s, M)}\]
and letting $t$ be a {\em distinguishing prime} relative to $N$ and
$M$ if $e(t,N) \ne e(t,M)$, where $e(t,N)$ is the largest exponent of
$t$ such that $t^{e(t,N)}$ divides $N$ and $e(t,M)$ the largest
exponent of $t$ such that $t^{e(t,M)}$ divides $M$.  

We formalize this definition and illustrate it with an example.

\begin{definition}\label{def:distinguishing}
Given two natural numbers $N \neq M$, we say any factor $t^{e(t, N)}$
of $N$ is a {\em distinguishing prime power relative to $M$} if
$e(t,N) \neq e(t, M)$ where $e(t,N)$ is the largest exponent of $t$
such that $t^{e(t, N)}$ divides $N$ and $e(t, M)$ is the largest
exponent of $t$ such that $t^{e(t, M)}$ divides $M$.  Similarly, we
say $t$ is a {\em distinguishing prime relative to $N, M$}.
\end{definition}

\begin{Example}
Let $N = 2 \cdot 3^2 \cdot 5^2$, $M = 2 \cdot 3^0 \cdot 7$.  Then
$t^{e(t, N)} = 3^2$ is a factor of $N$ where $2$ is the largest
exponent of $3$ such that $3^2$ divides $N$ and $e(3,M) = 0$ because
$0$ is the largest exponent of $3$ such that $3^0$ divides $M$.  We
can distinguish $N$ from $M$ by the fact that $3^2$ appears in the
prime factorization of $N$ but not in the prime factorization of $M$.
\end{Example}

In the proof of Theorem~\ref{thm:main} we need the following lemma.

\begin{notation}
If $f$ is a function and $(N_k)$ is a sequence of natural numbers, let
$f^k$ stand for the $k$-th recursive application of $f$ to an
arbitrary element of $(N_k)$. For example, if $n_5 \in (N_k)$, then
$f^0(n_5) = n_5$, $f^1(n_5) = f(n_5)= n_6$ and $f^3(n_5) = (f \circ f
\circ f)(n_5) = n_8$.
\end{notation}

\begin{lemma}\label{lem:lcm}
Let $N = AB$, where $A, B$ are coprime nontrivial factors of $N$.
Given a polynomial function $f\colon S \to S$, where $S$ is a finite nonempty
subset of the natural numbers, let $\lambda$ be the length of the
cycle of the sequence $(N_k)$ generated by the rule $n_{i+1} = f(n_i)
\bmod{N}$, for all $i \ge 1$, with a given first element $n_0 \in
S$. Then
 \[\lambda = \lcm(\lA, \lB),\]
where $\lA$ and $\lB$ are the lengths of the cycles of the
corresponding sequences obtained by reducing each element $n_k \in
(N_k)$ modulo $A$ and $B$, respectively.
\end{lemma}

\begin{proof}
Fix an element $x$ in the cycle of $(N_k)$.  By definition,
$\lambda$ is the least positive integer such that $f^{\lambda}(x)
\equiv x \bmod{N}$, which implies $f^{\lambda}(x) - x \equiv 0
\bmod{N}$.  In other words, $f^{\lambda}(x) - x$ is a multiple of
$N$.  Since $N = AB$, any multiple of $N$ is a multiple of both $A$
and $B$, that is, $f^{\lambda}(x) - x \equiv 0 \bmod{A, B}$,
implying $f^{\lambda}(x) \equiv x \bmod{A, B}$.  This shows
$\lambda$ is a common multiple of both $\lA$ and $\lB$.
We are left with showing $\lambda$ is the least such multiple.

By way of contradiction, suppose that $\lambda$ is not the least
common multiple of $\lA$ and $\lB$.  Then there is a
positive integer $\mu < \lambda$ such that $\mu$ is a multiple of both
$\lA$ and $\lB$.  This means $f^\mu(x) \equiv x \bmod{A,
  B}$, which implies $f^\mu(x) - x \equiv 0 \bmod{A, B}$.  But this
result implies $f^\mu(x) - x$ is a multiple of $N$ too, so $f^\mu(x) - x
\equiv 0 \bmod{N}$ implying $f^\mu(x) \equiv x \bmod{N}$.  In other
words, there is a positive integer $\mu < \lambda$ such that $f^\mu(x)
\equiv x \bmod{N}$, violating the hypothesis that $\lambda$ is the
least positive integer with the property.  Therefore, no such $\mu$
exists and $\lambda = \lcm(\lA, \lB)$, as desired.
\end{proof}

\begin{theorem}[name=A characterization of nontrivial collisions]\label{thm:main}
Let $N = AB$, where $A, B$ are coprime nontrivial factors of $N$.  Let
$f\colon S \to S$ be a polynomial function, where $S$ is a finite
nonempty subset of the natural numbers.  Let $(N_k)$ be the infinite
sequence generated by the rule $n_{i+1} = f(n_i) \bmod{N}$, for all $i
\ge 1$, with a given first element $n_0 \in S$.  Similarly, let
$(A_k)$ and $(B_k)$ be the infinite sequences generated by reducing
modulo $A, B$ each element $n_k \in (N_k)$.  Then $\lA \neq \lB$ if
and only if there exists a natural number $m < \lambda$ such that $m$
is a multiple of $\lA$ but $m$ is not a multiple of $\lB$, or {\em
  vice-versa}, where $\lambda$ is the length of cycle in the sequence
$(N_k)$ and $\lA, \lB$ are the lengths of the cycles of their
corresponding sequences. Moreover,
  \[1 < \gcd(f^m(x) - x, N) < N\]
for some element $x$ inside the cycle of $(N_k)$.
\end{theorem}

\begin{proof}
The converse implication is easily proved by noticing that if there is
a natural number $m < \lambda$ such that $m$ is a multiple of $\lA$
but $m$ is not a multiple of $\lB$, then $\lA \ne \lB$.  (If $m$ is a
multiple of $\lB$ but not a multiple of $\lA$, then again $\lA \ne
\lB$.)

Let us now prove the forward implication. We must show that (1) a
certain natural number $m$ exists and (2) that $m$ satisfies $1 <
\gcd(f^m(x) - x, N) < N$ for some fixed element $x$ in the cycle of
$(N_k)$ generated by $f$.  To prove the existence of $m$ we may take
either one of two paths, namely (a) show that there is a natural
number $m < \lambda$ such that $m$ is a multiple of $\lA$ but $m$ is
not a multiple of $\lB$ or (b) show that $m$ is not a multiple of
$\lA$ but $m$ is a multiple of $\lB$.  We will take path (a).

Let $\lA \neq \lB$. By Lemma~\ref{lem:lcm}, $\lambda = \lcm(\lA,
\lB)$.  Now write
  \[\lA = \prod_s s^{e(s, \lA)} \qquad\text{and}\qquad  \lB = \prod_s s^{e(s, \lB)},\]
where $s$ is a prime number and $e(s, \lA)$ represents the exponent of
$s$ in the prime factorization of $\lA$. Let $t$ be a distinguishing
prime relative to $\lA, \lB$ in the sense of
Definition~\ref{def:distinguishing}.  Without loss of generality,
assume $e(t, \lA) < e(t, \lB)$. Let $m = \lambda/t$.  By construction,
$m < \lambda$.  Given that $t$ is a distinguishing prime relative to
$\lA, \lB$ and $e(t, \lA) < e(t, \lB)$, then $e(t, m) = e(t, \lB) - 1$
because $m = \lambda/t$ and $\lambda = \lcm(\lA, \lB)$, implying $m$
is not a multiple of $\lB$.  In the prime factorization of $\lA$,
however, the largest possible value for $e(t, \lA)$ is $e(t, \lB) -
1$, thus $m$ is a multiple of $\lA$, as desired.  We are left with
showing $1 < \gcd(f^m(x) - x, N) < N$.

Fix an element $x$ in the cycle of $(N_k)$ generated by $f$.  By
definition, $\lambda$ is the least positive integer such that
$f^{\lambda}(x) \equiv x \bmod{N}$.  Since $m < \lambda$, then $f^m(x)
\not\equiv x \bmod{N}$, otherwise $m$ would be the length of the cycle
of $(N_k)$.  Thus, $f^m(x) - x \not\equiv 0 \bmod{N}$, meaning $f^m(x)
- x$ is not a multiple of $N$.  Since we have already established that
$m$ is a multiple of $\lA$, so $f^m(x) \equiv x \bmod{A}$ implying
$f^m(x) - x$ is a multiple of $A$.  Therefore, $f^m(x) - x$ is a
multiple of $A$ but not a multiple of $N$.  Hence,
  \[1 < \gcd(f^m(x) - x, N) < N,\]
as desired.
\end{proof}

As the proof of Theorem~\ref{thm:main} suggests, we can split $N$ by
finding a number $m$, if it exists, having the property that it is a
multiple of $\lA$ but not of $\lB$, or {\em vice-versa}.  Since the
\qpfa\ is able to compute $\lambda$ in polynomial-time, we are left
with finding $m$.  The classical version of Pollard's Rho searches for
a nontrivial collision among pairs of numbers in a cycle generated by
an iterated polynomial function.  Since Theorem~\ref{thm:main} gives
us a nontrivial collision if we find $m$, we can replace the classical
version's searching procedure with the \qpfa\ (which will give us
$\lambda$) followed by a search for $m$.  So, the algorithm we present
next is a quantum version of Pollard's Rho.

\section{A quantum version of Pollard's Rho}\label{sec:quantum-rho}

Since a quantum model of computation provides the polynomial-time
\qpfa~\cite[section~5.4.1, page~236]{qc}, we can design a quantum
version of Pollard's Rho.  If we just translate the classical version
in a trivial way, it will not be obvious how to take advantage of
quantum parallelism because, for example, the typical polynomial
$f(x)=x^2 + 1 \bmod{N}$ used in the classical version has no
easy-to-find closed-form formula that we can use.  We would end up
with an exponential quantum version of the strategy.  Not every
function used in the classical version will produce an efficient
quantum version of the method.

The straightforward way to take advantage of quantum parallelism is to
find a closed-form formula for the iterated function that can be
calculated in polynomial-time. (Informally, an arithmetical expression
is said to be in closed-form if it can be written in terms of a finite
number of familiar operations.  In particular, ellipses are not
allowed if they express a variable number of operations in the
expression~\cite[section~5.3, page~282]{hein1996}.)

Theorem~\ref{thm:closed} on page~\pageref{thm:closed} proves that the family of iterated functions %cm
%v
\begin{equation}\label{eq:fam}
  f(x) = ax^2 + bx + \frac{b^2 - 2b}{4a}
%e
\end{equation}
has closed-form formula %cm
%v
  \[f^n(x) = \frac{2\alpha^{2^n} - b}{2a},\]
%e
where $\alpha = (2ax + b)/2$.  Having the closed-form expression for
$f^n(x)$, we can take advantage of quantum parallelism.  However,
since we desire a polynomial-time algorithm, we must find a way to
keep the exponent small in $\alpha^{2^n}$, which is achievable if we
reduce $2^n$ modulo $\ord(\alpha, N)$ or modulo a multiple of
$\ord(\alpha, N)$.  We get
%v
  \[f^n(x) = (2\alpha^{\gamma} - b)(2a)^{-1} \bmod{N},\]
%e
where $\gamma = 2^n \bmod{r}$ and $r = \ord(\alpha, N)$ with
$\gcd(\alpha, N)=1$.  One last requirement is choosing $a$ and $b$
such that $f^n(x)$ is a polynomial of integer coefficients.  For
example, if we set $a = 1$, $b = 2$ and let $x_0$ be an initial value
for the sequence, we get the polynomial $n_{i+1} = f(x_i) = x_i^2 +
2x_i \bmod{N}$ whose closed-form formula is
  \[n_{i+1} = g(i) = ((x_0 + 1)^{\gamma} - 1) \bmod{N},\]
where $\gamma = {2^i \bmod r}$, $r = \ord(x_0 + 1, N)$ and $x_0$ is
some initial value such that $\gcd(x_0 + 1, N) = 1$.

\begin{Remark} 
The family of functions defined by Equation~\ref{eq:fam} is not the
only one that can be used.  For example, $f(x) = a x \bmod{N}$, where
$1 < a \bmod{N}$ is fixed, is a family of functions useful to the
method too.  In fact, this family reduces the quantum version of
Pollard's Rho to \sa.  We investigate this reduction in
Section~\ref{sec:restrict}.
\end{Remark}

\begin{Remark}
The dependency on $\ord(\alpha, N)$ makes the quadratic family of
Equation~\ref{eq:fam} an alternative factoring strategy when both
\sa~(Section~\ref{sec:shor}) and its extended version
(Section~\ref{sec:extended}) fail.  We can use $r$ computed by the
failed attempts to satisfy the closed-form formula $g$.
\end{Remark}

Although these polynomials $g$ of integer coefficients defined by
Equation~\ref{eq:fam} contain a cycle, they are not, in general,
periodic functions.  A periodic $g$ would require $g(x, i + \lambda) =
g(x, i)$ for every $x$ in the domain of $g$, for some $\lambda > 0$,
which is not satisfied by all $x$.  To get a periodic function for
taking advantage of the \qpfa, which is the method that will provide
us with $\lambda$, we can set its initial element $x_0$ to some
element in the cycle contained in $g$.  Since the length of the cycle
in $g$ is bounded by $N$, then $g(N)$ must be an element in the cycle,
so we get the desired restriction on $g$.  Therefore, another key step
in the algorithm is to compute the $N$-th element of the sequence
generated by $g$.

%% --8<---------------cut here---------------start------------->8---
%% New Kowada's additions %%%%%%%%%%%%%%%%%%%%%%%%%%%%%%%%%%%%%%%%%%
%%
The function which we should use in the quantum version of Pollard's
Rho is $g(i) = f^i(x_0)$, where $i$ is the $i$-th element in the
sequence generated by $g$.  As an example, we use 
  \[g(i)= {2\alpha^{2^i} - b}({2a})^{-1}\bmod{N},\]
where $\alpha = (2ax_0 + b)2^{-1}\bmod{N}$, which corresponds to the
iterated function
  \[f(x) = ax^2 + bx + ({b^2 - 2b})({4a})^{-1}\bmod{N}.\] 
Observe that the strategy does not need $f$.  The procedure {\sc
  quantum-rho}, expressed in Algorithm~\ref{alg:quantum-rho}, uses
only $g$.  The important requirement is for $g$ to correspond to an
iterated function.  For example, instead of the $g$ we use, we could
take $g(i) = x_0 a^i \bmod{N}$ because this family corresponds to the
family $f(x) = ax \bmod{N}$ of iterated functions.  The importance of
iterated functions is the same as that of polynomials of integer
coefficients in the classical version of Pollard's Rho
(Theorem~\ref{thm:poly-importance},
Section~\ref{sec:pollard-strategy}).
%% --8<---------------cut here---------------end--------------->8---

\begin{algorithm}[htb]
\def\spacing{0.2cm}
\caption{A quantum version of Pollard's Rho using an integer periodic
  sequence modulo $N$ generated by a closed-form formula $g$
  corresponding to an iterated function.  Assume $N$ is not a prime
  power.}
\label{alg:quantum-rho}
\begin{algorithmic}[0]
\Procedure{quantum-rho}{$N$}
  %% \State $\alpha \gets (ax_0 + 2^{-1}b) \bmod{N}$
  %% \State $r_\alpha \gets \textsc{quantum-order-finding}(\alpha, N)$
  %% \State $x_\N \gets g(N)$ \Comment The $N$-th element of the
  %% sequence generated by $g$.
  \State $r_\g \gets \textsc{quantum-period-finding}(g)$
  \For{$d \text{ in } \text{divisors}(r_\g)$}
    \State $m \gets \gcd(g(N + r_\g / d) - g(N), N)$
    \If{$1 < m < N$}
      \State\Return $m$ \Comment Nontrivial collision found.
    \EndIf
  \EndFor
  \State\Return none
\EndProcedure\vskip\spacing

%% \Function{$f$}{$x, a, b, N$}
%%   \State\Return $(ax^2 + bx + (b^2 - 2b)(4a)^{-1}) \bmod{N}$
%% \EndFunction\vskip\spacing

%% \Function{$g$}{$x, i, a, b, r_\alpha, N$}
%%   \State $\alpha \gets (ax + 2^{-1}b) \bmod{N}$
%%   \State\Return $(2\alpha^{\gamma} - b)(2a)^{-1} \bmod{N}$  where $\gamma = 2^i \bmod{r_\alpha}$ \Comment Closed-form formula for $f$.
%% \EndFunction\vskip\spacing
\end{algorithmic}
\end{algorithm}

%% \begin{algorithm}[htb]
%% \def\spacing{0.2cm}
%% \caption{The divisors generator.}
%% \label{alg:divisors}
%% \begin{algorithmic}[0]
%% \Procedure{$\divisors$}{$r,N$}
%%   \For{$i = 1 \text{ to } i = \text{bits}(N)$}
%%     \State $p \gets \text{prime}(i)$ \Comment{The $i$-th smallest prime.}
%%     \If{$p$ divides $r$}
%%       \State\textbf{yield} $p$
%%     \EndIf
%%   \EndFor
%%   \State\Return none
%% \EndProcedure\vskip\spacing

%% \Function{$\bits$}{$N$}
%%   \State\Return $\floorOP(\lg(N)) + 1$ \Comment The number of bits in $N$.
%% \EndFunction
%% \end{algorithmic}
%% \end{algorithm}

We now describe the steps of \textsc{quantum-rho}, the procedure
expressed in Algorithm~\ref{alg:quantum-rho}.  In
Section~\ref{sec:circuit}, we give a description of a quantum circuit
that could be used to execute Algorithm~\ref{alg:quantum-rho} on a
quantum computer.

The procedure consumes $N$, the composite we wish to factor.  At a
first stage, Algorithm~\ref{alg:quantum-rho} must use a circuit like
the one described in Section~\ref{sec:circuit} to compute the length
$r_\g$ of the cycle contained in the sequence generated by the
function $g$.  Then the procedure checks to see if any pair $(N, N +
r_\g/d)$ is a nontrivial collision relative to $(N_k)$.  The
characterization of nontrivial collisions established by
Theorem~\ref{thm:main} asserts that if $r_\g/d$ is a multiple of $\lA$
but not a multiple of $\lB$, then $(N, N + r_\g/d)$ is a nontrivial
collision, where $\lA$ and $\lB$ are the lengths of the cycles of the
sequences generated by $g$ reduced modulo $A$ and $B$, respectively,
where $N = AB$ and $A$ is coprime to $B$.

An example should clarify the procedure.

\begin{Example}
Let $p = 7907$, $q = 7919$ so that $N = pq = 62615533$.  Let $a = 1$,
$b = 2$ so that $f$ is the polynomial of integer coefficients $x_{i+1}
= f(x_i) = x_i^2 + 2x_i \bmod{N}$ and choose $x_0 = 3$ as its initial
value.  The closed-form formula for $f$ is $g(i) = ((x_0 + 1)^{\gamma}
- 1) \bmod{N}$, where $\gamma = {2^i \bmod r}$ and $r = \ord(x_0 + 1,
N)$.  Assuming we have $r = \ord(4,N) = 15649927$, we compute $x_\N =
g(N) = 10689696$, the $N$-th element of the sequence generated by $g$,
which is an element in the cycle.  Restricting $g$ by letting its
initial value be $x_\N$, we get a periodic function whose period $r_\g
= 608652$ is computed by the \qpfa.  The procedure then looks for a
nontrivial collision by trying pairs $(N, N + r_\g/d)$ for prime
divisors $d$ of $r_\g$.  Since $608652$ is even, $d = 2$ is the first
prime divisor of $r_\g$ revealed.  In this case, the algorithm finds a
nontrivial collision using $d=2$ because the prime factorizations of
\begin{align*}
     r_\g/2 &= 2 \times 3^2  \times 11 \times 29 \times 53\\
  \lambda_P &= 2   \times 3   \times 11 \times 29\\
  \lambda_Q &= 2^2 \times 3^2  \times 53
\end{align*}
reveals that $r_\g/2$ is a multiple of $\lambda_P$ but not of
$\lambda_Q$, hence $(N, N + r_\g/2)$ is a nontrivial collision, which
we confirm by computing $\gcd(16896691 - 10689696, 62615533) = 7907$.
\end{Example}

The procedure ``divisors'' used in Algorithm~\ref{alg:quantum-rho}
searches for any small prime divisors $d$ of $r_\g$ if it can find.
%% Aiming for a clear exposition, we express the procedure as a
%% generator\footnote{Being a generator means that as soon as the first
%%   prime divisor of $r$ is found the \textbf{yield} keyword interrupts
%%   the execution of the procedure, immediately providing the caller
%%   with the first prime divisor.  On subsequent invocations, the
%%   procedure continues from where it left off, providing the caller
%%   with the next smallest prime divisor.  Generators were conceived as
%%   early as 1977 and presented in the implementation of
%%   CLU~\cite{liskov1977}. Many popular programming languages currently
%%   implement generators.  A generator is not a required feature of the
%%   programming language used to implement the algorithm.  We use it
%%   with the objective of a clear exposition.} in
%% Algorithm~\ref{alg:divisors}.
We can guarantee a polynomial-time bound for the procedure by
restricting the search up to the $n$-th smallest prime, where $n$ is
the number of bits in $N$.

Asymptotically, the complexity of Algorithm~\ref{alg:quantum-rho} is
the same as that of Shor's algorithm, $O(\lg^3 N)$, given that the
procedure ``divisors'' is interrupted after $n$ attempts, where $n$ is
the number of bits in $N$.  (Finer estimates for \sa\ have been
given~\cite{beauregard2002,beckman1996}.)

We close this section by proving the equivalence between the family of
quadratic polynomials and their closed-form formulas used by
Algorithm~\ref{alg:quantum-rho}.

\begin{theorem}\label{thm:closed}
The closed-form formula for the iterated family %cm
%v
  \[f(x) = ax^2 + bx + \frac{b^2 - 2b}{4a}\]
%e
of such quadratic polynomials is %cm
%v
  \[f^i(x_0) = \frac{2\alpha^{2^i} - b}{2a},\]
%e
where $\alpha = (2ax_0 + b)/2$ and $f^i$ stands for the $i$-th
iteration of $f$ having $x = x_0$ as an initial value.
\end{theorem}

\begin{proof}
We prove by induction on $i$.  The first element of the sequence
generated by $f$ is $x_0$ by definition, which we verify by computing %align by =
%v
%% \begin{align*}
%%   f^0(x_0) &= x_0\\
%%            &= \frac{2ax_0 + b - b}{2a}\\
%%            &= \frac{2ax_0 + b}{2a} - \frac{b}{2a}\\
%%            &= 2 \left(\frac{[2ax_0 + b]/2}{2a}\right) - \frac{b}{2a}\\
%%            &= \frac{2\alpha - b}{2a}\\
%%            &= \frac{2\alpha^{2^0} - b}{2a}.
%% \end{align*}
\begin{equation*}
  f^0(x_0) = x_0
           = \frac{2ax_0 + b - b}{2a}
           %% = \frac{2ax_0 + b}{2a} - \frac{b}{2a}
           = 2 \left(\frac{[2ax_0 + b]/2}{2a}\right) - \frac{b}{2a}
           = \frac{2\alpha - b}{2a}
           = \frac{2\alpha^{2^0} - b}{2a}.
\end{equation*}
%e
%% We also verify that $x_1 = f^1(x_0)$ 
%% by computing %align by =
%% %v
%% \begin{align*}
%%   f^1(x_0)&= ax_0^2 + bx_0 + \frac{b^2 - 2b}{4a}\\
%%         &= \frac{a^2x_0^2 + abx_0 + (b^2 - 2b)/4}{a}\\
%%         &= \frac{a^2x_0^2 + abx_0 + b^2/4 - b/2}{a}\\
%%         &= \frac{(ax_0 + b/2)^2 - b/2}{a}\\
%%         &= \frac{(2ax_0 + b)/2)^2 - b/2}{a}\\
%%         &= \frac{\alpha^2 - b/2}{a}\\
%%         &= \frac{2\alpha^2 - b}{2a}.
%% \end{align*}
%e
Now, suppose
\[f^k(x_0) = \frac{2\alpha^{2^k} - b}{2a},\]  
 for some $k \ge 0$ where $\alpha = (2ax_0 + b)/2$.%%   We must show %cm
%% %v
%%   \[f^{k+1}(x_0) = \frac{2\alpha^{2^{k+1}} - b}{2a}.\]
%% %e

Since \[f^{k+1}(x_0) = f(f^k(x_0)),\] we deduce %align by =
%v
\begin{align*}
  f^{k+1}(x_0) &= f\Biggl(\frac{2\alpha^{2^k} - b}{2a}\Biggr)\\
  &= a\left(\frac{2\alpha^{2^k} - b}{2a}\right)^2 + b\left(\frac{2\alpha^{2^k} - b}{2a}\right) + \frac{b^2 - 2b}{4a}\\
  %% &= a\left(\frac{4\alpha^{2^{k+1}} -4\alpha^{2^k}b  + b^2}{4a^2}\right) + \frac{2\alpha^{2^k}b - b^2}{2a} + \frac{b^2/2 - b}{2a}\\
  &= \frac{4\alpha^{2^{k+1}} -4\alpha^{2^k}b  + b^2}{4a} + \frac{2\alpha^{2^k}b - b^2}{2a} + \frac{b^2/2 - b}{2a}\\
  %% &= \frac{4\alpha^{2^{k+1}} -4\alpha^{2^k}b  + b^2}{4a} + \frac{2\alpha^{2^k}b - b^2}{2a} + \frac{b^2/2 - b}{2a}\\
  %% &= \frac{2\alpha^{2^{k+1}} -2\alpha^{2^k}b  + b^2/2}{2a} + \frac{2\alpha^{2^k}b - b^2}{2a} + \frac{b^2/2 - b}{2a}\\
  %% &= \frac{2\alpha^{2^{k+1}} -\cancel{2\alpha^{2^k}b}  + b^2/2 + \cancel{2\alpha^{2^k}b} - b^2 + b^2/2 - b}{2a}\\
  &= \frac{2\alpha^{2^{k+1}} -2\alpha^{2^k}b  + b^2/2 + 2\alpha^{2^k}b - b^2 + b^2/2 - b}{2a}\\
  %% &= \frac{2\alpha^{2^{k+1}} + \cancel{b^2} - \cancel{b^2} - b}{2a}\\
  &= \frac{2\alpha^{2^{k+1}} - b}{2a},
\end{align*}
as desired.
\end{proof}

%% --8<---------------cut here---------------start------------->8---
%% The family of polynomial iterated functions used in
%% Algorithm~\ref{alg:quantum-rho} requires computing an exponentiation
%% whose exponent grows exponentially, but we were able to keep this
%% exponent small by reducing it modulo $\ord(\alpha, N)$, which in turn
%% requires us to invoke the \qofa.  This requirement emerges due to the
%% closed-form formula of the family of polynomial iterated functions we
%% used in Algorithm~\ref{alg:quantum-rho}.  Although the quantum version
%% of Pollard's Rho does not seem to offer a performance gain when
%% compared to \sa\ or its extended version, it does bring us an
%% understanding about \sa\ we did not have before.
%% --8<---------------cut here---------------end--------------->8---

\section{The strategy in Shor's algorithm}\label{sec:shor}

Shor's algorithm~\cite{shor94,shor97,shor99} brought quantum computing
to the spotlight in 1994 with its exponential speed up of a solution
for the problem of finding the order of an element $x$ in a finite
group.  The order $r = \ord(x, N)$ of an element $x \in \Z$ is the
smallest positive number $r$ such that $x^r = 1 \bmod{N}$.  There is a
polynomial-time reduction of the problem of factoring to the problem
of finding the order of an element~\cite{miller1976}. By computing
$\ord(x, N)$ in polynomial-time in a quantum model of computation, the
rest of the work of factoring $N$ can be carried out efficiently in a
classical model of computation, providing us with an efficient
solution to the problem of factoring.

\begin{algorithm}[htb]
\def\spacing{0.2cm}
\caption{Shor's algorithm using $x \bmod{N}$ with $x$ coprime to $N$
  and $N$ not a prime power.}
\label{alg:shor}
\begin{algorithmic}[0]
\Procedure{shor}{$x, N$}
  \State $r \gets \textsc{quantum-order-finding}(x, N)$
  \If{$r$ is even}
    \State $p \gets \gcd(x^{r/2} - 1, N)$
    \If{$1 < p < N$}
      \State\Return $p$ \Comment Shor's condition is satisfied.
    \EndIf
  \EndIf
  \State\Return none
\EndProcedure\vskip\spacing
\end{algorithmic}
\end{algorithm}

If $r = \ord(x, N)$ is even, for the chosen $x \in \Z$, \sa\ finds a
factor by computing $\gcd(x^{r/2} -1, N)$.  The essence of the
strategy comes from the fact that
\begin{equation}\label{eq:shor-again}
 (x^{r/2} -1) (x^{r/2} + 1) = x^{r} - 1 = 0 \bmod{N}.
\end{equation}
It is easy to see that if $x^{r/2} \equiv -1 \bmod{N}$ then the
equation is trivially true, leading the computation of the $\gcd$ to
reveal the undesirable trivial factor $N$.  In other words, the
algorithm fails.  \sa\ needs not only an even order, but also $x^{r/2}
\not\equiv -1 \bmod{N}$.  

If $r$ is odd, an extension~\cite{extension,johnston2017} of
\sa\ is useful for further attempts at splitting $N$.

In the expression of \sa\ in Algorithm~\ref{alg:shor}, we assume $N$
is not a prime power~\cite[section~6, page~130]{shor94}.  Verifying a
number is not a prime power can be done
efficiently~\cite{aks2004}\cite[chapter~3, note~3.6, page~89]{hac} in
a classical model of computation.  Througout this document, whenever
\sa\ is mentioned, we assume such verification is applied.

\section{An extension of Shor's algorithm to odd orders}\label{sec:extended}

Special ways of using odd orders have been know for some
time~\cite{cao2005,lawson2015,xu2018} and a general extension of
\sa\ to odd orders has also been presented~\cite{johnston2017}.  More
recently, we presented a different
perspective~\cite[section~3]{extension} of the same
result~\cite{johnston2017} to extend \sa\ to any odd order,
establishing the equivalence of both perspectives.

\begin{algorithm}[htb]
\def\spacing{0.2cm}
\caption{An extension of Shor's algorithm in which a certain fixed
  number of small divisors of the order of $x$ in $\Z$ is considered.
  The procedure assumes $N$ is not a prime power.}
\label{alg:extended}
\begin{algorithmic}[0]
\Procedure{extended-shor}{$x, N$}
  \State $r \gets \textsc{quantum-order-finding}(x, N)$
  \For{$d \text{ in } \text{divisors}(r)$}
    \State $p \gets \gcd(x^{r/d} - 1, N)$
    \If{$1 < p < N$}
      \State\Return $p$
    \EndIf
  \EndFor
  \State\Return none %%\Comment Try a different $x$
\EndProcedure\vskip\spacing
%% \Function{$g$}{$i$}
%%   \State\Return $x^i \bmod{N}$
%% \EndFunction
\end{algorithmic}
%% The procedure ``divisors'' used in Algorithm~\ref{alg:extended} is
%% expressed in Algorithm~\ref{alg:divisors}.
\end{algorithm}

The reason Shor's procedure needs $r$ to be even is due to
Equation~(\ref{eq:shor-again}), but a generalization of this equation
naturally leads us to the extended version
(Algorithm~\ref{alg:extended}) showing how any divisor $d$ of $r$ can
be used. For instance, if $3$ divides $r$, then
 $(x^{r/3} -1) (1 + x^{r/3} + x^{2r/3}) = x^{r} - 1 \equiv 0 \bmod{N}$.
In general,
\begin{equation}
\label{eq:extended-shor}
 (x^{r/d} -1) \left(\sum_{i=0}^{d - 1} x^{ir/d}\right) = x^{r} - 1 \equiv 0 \bmod{N},
\end{equation}
whenever $d$ divides $r$.

However, the extended version must hope that $x^{r/d} \not\equiv -1
\bmod{N}$, when the extended version would also fail, as it similarly
happens in the original version of \sa.
Equation~(\ref{eq:extended-shor}) gives us little understanding of
when such cases occur, but Theorem~\ref{thm:main} provides deeper
insight.  Sufficient conditions~\cite[sections~2--3,
  pages~2--3]{johnston2017} for the success of the extended version
and an equivalent result~\cite[section~3]{extension} have been
presented.

\section{Pollard's Rho is a generalization of \sa}\label{sec:restrict}

%% We need to show how Shor's algorithm is naturally a restriction of
%% Pollard's Rho in its quantum version.  The key to see this is to
%% notice that Pollard's Rho needs the length (or a multiple of it) L
%% of the cycle.  Having L, it can search for a suitable m < L that
%% provides a nontrivial collision.

%% This is the same that Shor's algorithm does.  It finds the length L
%% of the cycle of some sequence (b^i mod N) and hopes m = L/2 is
%% suitable for a nontrivial collision.

%% Consider the quantum version of Pollard's Rho restricted to
%% periodic functions such as b^i mod N.  Computing the periods is the
%% same as computing r = ord(b, N).  Now let m = r/2, assuming r is
%% even, and check whether (x^m, 1) is a nontrivial collision.  This
%% is exactly Shor's algorithm.  Therefore, Shor's algorithm is a
%% restricted quantum version of Pollard's Rho.

As we promised in Section~\ref{sec:quantum-rho}, let us revist the
family of functions
\begin{equation}\label{eq:fam:shor}
  f(x) = a x \bmod{N},
\end{equation}
where $1 < a \bmod{N}$ is a fixed natural number.  Let $x_0 = 1$ be
the first element of the sequence generated by this iterated function.
If we choose $a = 2$, we get the sequence $1, 2, 4, 8, ...$.  This
iterated function has closed-form formula
\begin{equation}\label{eq:g:shor}
  g(i) = a^i \bmod{N}.
\end{equation}
It follows that $g$ is purely periodic and the length of its period is
$\ord(a, N)$.  Since $g$ is purely periodic and $1$ is always an
element of the sequence, instead of using $g(N)$ as an element in the
cycle as we did in Algorithm~\ref{alg:quantum-rho}, we use $g(0) = 1$.
Using $g$ (from Equation~\ref{eq:g:shor}) for the quantum version of
Pollard's Rho, we get Algorithm~\ref{alg:quantum-rho-shor}.  The
difference between Algorithm~\ref{alg:quantum-rho-shor} and
Algorithm~\ref{alg:quantum-rho} is the choice of the function $g$ and
the choice of an element that we can be sure it belongs to the cycle.

\begin{algorithm}[htb]
\def\spacing{0.2cm}
\caption{A quantum version of Pollard's Rho using an integer periodic
  sequence modulo $N$ generated by a closed-form formula $g$ from the
  family defined by Equation~\ref{eq:fam:shor} on
  page~\pageref{eq:fam:shor} with first element $x_0 = 1$. Assume $N$
  is not a prime power.}
\label{alg:quantum-rho-shor}
\begin{algorithmic}[0]
\Procedure{quantum-rho'}{$a, N$}
  \State $r_\g \gets \textsc{quantum-period-finding}(g)$
  \For{$d \text{ in } \text{divisors}(r_\g)$}
    \State $p \gets \gcd(g(r_\g / d) - g(0), N)$ \Comment Notice $g(0)=1$.
    \If{$1 < p < N$}
      \State\Return $p$ \Comment Nontrivial collision found.
    \EndIf
  \EndFor
  \State\Return none
\EndProcedure\vskip\spacing

\Function{$g$}{$i$}
  \State\Return $a^i \bmod{N}$
\EndFunction
\end{algorithmic}
\end{algorithm}

Let us see an example of the steps of Algorithm~\ref{alg:quantum-rho-shor}.

\begin{Example}
Let $a = 3$ so that the function defined by Equation~\ref{eq:fam:shor}
has closed-form formula $g(i) = 3^i \bmod{N}$.  Let $p = 19$ and $q =
11$ so that $N = pq = 209$.  The period of $g$ is $r_\g = 90$ and so
$d = 2$ divides $r_\g$.  Since $\ord(3,11) = 5 \neq 18 = \ord(3,19)$
and $r_\g/d = 45$ is a multiple of $\ord(3,11)$ but not a multiple of
$\ord(3,19)$, Theorem~\ref{thm:main} guarantees that $(0, 45)$ is a
nontrivial collision.  We check the result computing $\gcd(g(90/2) -
g(0)) = \gcd(3^{90/2} - 1, 209) = \gcd(55, 209) = 11$, as desired.
%% >>> [pow(3,i, N) for i in range(90)].index(56)
%% 45
\end{Example}

Since $r_\g$ happens to be even, we can see that the example follows
the exact steps of the original algorithm published by Peter Shor in
1994.

We end this section with one final example.

\begin{Example}
Let $p = 7907$, $q = 7919$ so that $N = pq = 62615533$.  If we pick $a
= 3$, we get $r_\g = \ord(3,N) = 15649927$, an odd integer, a case in
which Shor's original algorithm would not succeed. Since $N$ has 26
bits, the procedure checks if any of the smallest 26 primes divides
$r_\g$.  The smallest eleven primes do not, but the twelfth prime is
37 and it divides $r_\g$, so Algorithm~\ref{alg:quantum-rho-shor}
finds a nontrivial factor by computing $\gcd(x^{r_\g/37} - 1, N) =
\gcd(48604330 - 1, 62615533) = 7907$.  We can see why it succeeds by
looking at the prime factorizations of
\begin{align*}
          r_\g/37  &= 59 \times 67 \times 107\\
  r_p &= 59 \times 67\\
  r_q &= 37 \times 107,
\end{align*}
where $r_p = \ord(3, 7907)$ and $r_q = \ord(3, 7919)$.  We see that
$r_\g/37$ is a multiple of $r_p$ but not a multiple of $r_q$, so
Theorem~\ref{thm:main} guarantees that $1 < \gcd(x^{r/37} - 1, N) <
N$.  From the point of view of the extended version of Shor's
algorithm, it succeeds because $x^{r/37} \not\equiv -1 \bmod{N}$, but
in the light of Theorem~\ref{thm:main} we get the deeper insight that
the strategy succeeds because 37 happens to be a distinguishing prime
relative to $r_p, r_q$ in the sense of
Definition~\ref{def:distinguishing}.
\end{Example}

\section{A description of a quantum circuit for Pollard's Rho}\label{sec:circuit}
%%%%%%%%%%%%%%%%%%%%%%%%%%%%%%%%%%%%%%%%%%%%%%%%%%%%%%%%%%%%%%%
%% These commands apply only to this section of the paper.   %%
%%%%%%%%%%%%%%%%%%%%%%%%%%%%%%%%%%%%%%%%%%%%%%%%%%%%%%%%%%%%%%%
\renewcommand{\N}{143}
\newcommand{\n}{8}  % bits necessários para representar N
\renewcommand{\r}{15}
\newcommand{\iz}{3} %i0
\newcommand{\xiz}{126} % (x0+1)^\iz mod N
\renewcommand{\ll}{15}
\renewcommand{\L}{32768}
\newcommand{\LL}{32767} % L-1
%%%%%%%%%%%%%%%%%%%%%%%%%%%%%%%%%%%%%%%%%%%%%%%%%%%%%%%%%%%%%%%

We now describe a circuit for the quantum version of Pollard's Rho
using elementary quantum gates.  For greater clarity, we implement the
circuit relative to the function $f(x) = x^2 + 2x \bmod{N}$ and take
$N = 11 \times 13$ as a concrete example.  Despite this particular
choice of $N$ in our description, the circuit is general for
the function $f(x) = x^2 + 2x \bmod{N}$ and describing different
functions would follow similar steps.  With this choice of $f(x)$, 
we have chosen $a = 1$ and $b = 2$ in Equation~\ref{eq:fam}, so $\alpha \equiv 2
\bmod{N}$ and
   \[g(i)=(x_0+1)^{2^i\bmod{r}} - 1 \bmod{N},\] 
where $r = \ord(x_0 + 1,N)$ and $x_0$ is some initial value such that
$\gcd(x_0 + 1, N) = 1$.  Let us let $x_0 = 2$ so that $\ord(3, \N) =
\r$.

The need for calculating $r$ implies that, before using this circuit,
we should see if Shor's original algorithm (or its extended version)
is able to split $N$.  If neither succeeds, then instead of running
either one of them again, the quantum version of Pollard's Rho using
the family of Equation~\ref{eq:fam} is an alternative, since the
number $r$ it needs is already computed by the failed attempts of
Shor's original algorithm and its extended version.

We illustrate first the operator for modular exponentiation
(Figure~\ref{fig:opExMod}).  The operator $U$ for calculating 
$g(i)$ is defined as 
$U \qr{i}\qr{y}\rightarrow\qr{i}\qr{y \oplus g(i)}$.
In our example,
  \[U\qr{i}\qr{y}\rightarrow\qr{i}\qr{y \oplus (3^{2^i \bmod{\r}} - 1 \bmod{\N})}.\]
The circuit for $U$ is illustrated by Figure~\ref{fig:opU}.  The
operators used in $U$ are the modular exponentiation operator,
described by Figure~\ref{fig:opExMod} and the modular SUB operator,
both of which are well-known operators~\cite{meter,vedral}.

\begin{figure}[htb]
\centering
\includegraphics[scale=0.6]{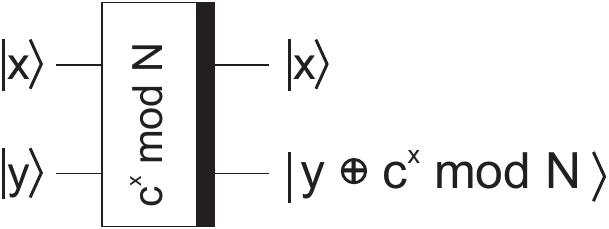}
\caption{Operator for modular exponentiation}
\label{fig:opExMod}
\end{figure}

\begin{figure}[htb]
\centering
\includegraphics[scale=0.7]{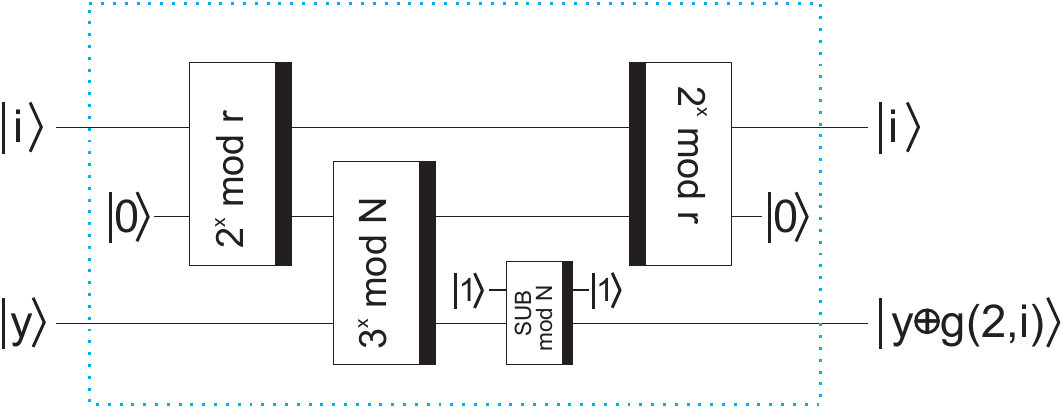}
\caption{The circuit for $U$ using $x_0 = 2$.}
\label{fig:opU}
\end{figure}

Let us now describe the steps in the \qpfa, illustrated by
Figure~\ref{fig:PeriodFinding}.
\begin{figure}[htb]
\centering
\includegraphics[scale=0.7]{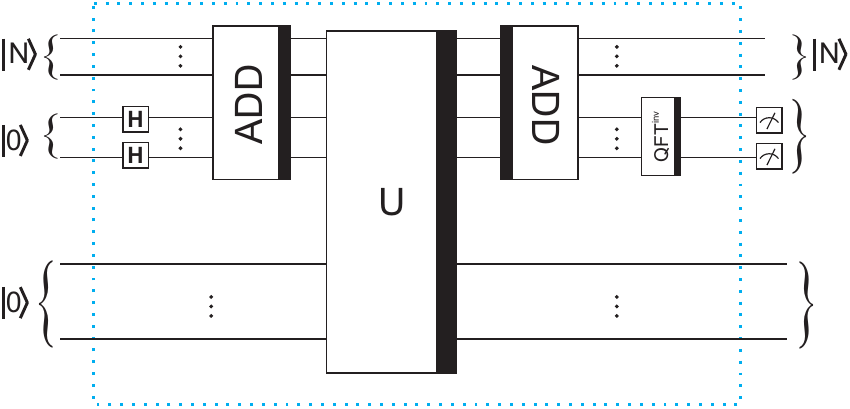}
\caption{Perioding-finding using the $U$ operator of
  Figure~\ref{fig:opU} to factor $N$.}
\label{fig:PeriodFinding}
\end{figure}
The initial state of the system is
\[ \qr{\Psi_0}=\qr{N}_n\qr{0}_{\ell}\qr{0}_n,\]
where $n=\lfloor \log_2 N\rfloor +1$ is the number of bits needed to
represent $N$ and $2^\ell$ is the number of elements evaluated by the
operator QFT.  In general, the size $\ell$ of the second register
satisfies $N^2 \le 2^\ell < N^2 + 1$.  See
Shor~\cite{shor1997polynomial} for more details.  The ancilla bits are
not shown in Figure~\ref{fig:PeriodFinding}.

After the Hadamard gates are applied, we get
\[ \qr{\Psi_1}=\qr{\N}_n H^{\otimes\ell}\qr{0}_\ell\qr{0}_n=\frac{1}{2^{\ell/2}}\sum_{i=0}^{2^\ell-1}\qr{\N}_n\qr{i}_\ell\qr{0}_n.\]
Our strategy to restrict $g$ to the cycle, before $U$ is applied, is
to give $g$ an initial value that is an element in the cycle, so we
use the ADDER operator to shift the register $N$ units ahead, yielding
\[\qr{\Psi_2}=\frac{1}{2^{\ell/2}}\sum_{i=0}^{2^\ell-1}\qr{\N}_n\qr{\N+i}_{\ell}\qr{0}_n.\]
The ADDER operator is defined as
$ADD\colon\qr{A}\qr{B}\rightarrow\qr{A}\qr{A+B}$. The size of the
second register needs to be greater than or equal to the first
register in this operator. See Vedral, Barenco and Eckert for more
information about the implementation of the ADDER~\cite{vedral}.  An
extra q-bit $\qr{0}$ is needed to avoid a possible overflow. In our
case, this extra q-bit is part of the second register of the ADDER,
but the Hadamard gate is not applied to this q-bit. For simplicity,
this bit is omitted in Figure~\ref{fig:PeriodFinding} and in the
description of the states of the system.

The next step is the application of $U$, after which the state of the
system is
\[\qr{\Psi_3}=U\qr{\Psi_2}=\frac{1}{2^{\ell/2}}\qr{\N}_n\sum_{i=0}^{2^\ell-1}\qr{\N+i}_{\ell}\qr{3^{2^{\N + i} \bmod{\r}}-1 \bmod{\N}}_n.\]

For $N=\N$, we have $2^{\N}\equiv 2^\iz \bmod{\r}$ and $n=\n$. Taking
$\ell=\lfloor\log_2(N^2)\rfloor+1=\ll$, we get
\[\qr{\Psi_3}=\frac{1}{\sqrt{\L}}\qr{\N}\sum_{i=0}^{\LL}\qr{\N+i}\qr{3^{2^{\iz+i} \bmod{\r}} - 1 \bmod{\N}}.\]
%
%Esta equação só funciona para N=143, L=32768,...
We can rewrite $\qr{\Psi_3}$ as
\[
\begin{array}{rlll}
\qr{\Psi_3}=\frac{1}{\sqrt{\L}}\qr{\N}\bigl(
&(\qr{143}+\qr{147}+\qr{151}+\qr{155}+\cdots+\qr{32907})&\qr{125}&+\\
&(\qr{144}+\qr{148}+\qr{152}+\qr{156}+\cdots+\qr{32908})&\qr{2}&+\\
&(\qr{145}+\qr{149}+\qr{153}+\qr{157}+\cdots+\qr{32909})&\qr{8}&+\\
&(\qr{146}+\qr{150}+\qr{154}+\qr{158}+\cdots+\qr{32910})&\qr{80}&\bigr).
\end{array}
\]

The next step is the application of the reverse of the ADDER operator,
after which we get
\[
\begin{array}{rlll}
\qr{\Psi_4}=\frac{1}{\sqrt{\L}}\qr{\N}\bigl(
&(\qr{0}+\qr{4}+\qr{8}+\qr{12}+\cdots+\qr{32764})&\qr{125}&+\\
&(\qr{1}+\qr{5}+\qr{9}+\qr{13}+\cdots+\qr{32765})&\qr{2}&+\\
&(\qr{2}+\qr{6}+\qr{10}+\qr{14}+\cdots+\qr{32766})&\qr{8}&+\\
&(\qr{3}+\qr{7}+\qr{11}+\qr{15}+\cdots+\qr{32767})&\qr{80}&\bigr).
\end{array}
\]

The final state before measurement is
$\qr{\Psi_5}=QFT^\dagger\qr{\Psi_4}$, that is,

\[\qr{\Psi_5}=\frac{1}{\L}\qr{\N}\sum_{i=0}^{\LL}\sum_{k=0}^{\LL}e^{2\pi \iota ik/\L }\qr{i}\qr{3^{2^{\iz+i} \bmod{\r}}-1\bmod{\N}},\]

where $\iota = \sqrt{-1}$.

We can use the principle of implicit measurement~\cite[box~5.4,
  page~235]{nielsen2004} to assume that the second register was
measured, giving us a random result from \{2, 8, 80, 125\}.  So, in
this example, there are four possible outcomes of the measurement in
the first register of the state $\qr{\Psi_5}$, all of which have the
same probability of being measured.  We might get either 0, 8192,
16384 or 24576, that is, $0/2^\ell$, $r_\g/2^\ell$, $2r_\g/2^\ell$ or
$3r_\g/2^\ell$, where $r_\g=4$ is the period of the cycle produced by
the function $g$. For more information on extracting $r_\g$ from the
measurement of $\qr{\Psi_5}$, please see
Shor~\cite{shor1997polynomial} and Nielsen, Chuang~\cite[section~5.3.1,
  page~226]{nielsen2004}.

Finally, we compute $m = \gcd(g(N + r_\g/2) - g(N), N) = \gcd(8 - 125,
143) = 13$, as desired.

\section{Conclusions}

John M.~Pollard presented in 1975 an exponential algorithm for
factoring integers that essentially searches the cycle of a sequence
of natural numbers looking for a certain pair of numbers from which we
can extract a nontrivial factor of the composite we wish to
factor, assuming such pair exists.  Until now there was no
characterization of the pairs that yield a nontrivial factor.  A
characterization is now given by Theorem~\ref{thm:main} in
Section~\ref{sec:main}.

Peter W.~Shor presented in 1994 a polynomial-time algorithm for
factoring integers on a quantum computer that essentially computes the
order $r$ of a number $x$ in the multiplicative group $\Z$ and uses
$r$ to find a nontrivial divisor of the composite.  No publication had
so far reported that Shor's strategy is essentially a particular case
of Pollard's strategy.

In the light of Theorem~\ref{thm:main} in Section~\ref{sec:main}, we
wrote a quantum version of Pollard's Rho
(Algorithm~\ref{alg:quantum-rho}, Section~\ref{sec:quantum-rho}) and
Section~\ref{sec:restrict} exposes the fact that by choosing a certain
family of functions the steps of the algorithm are reduced to the
exact steps of \sa.

\printbibliography 

@article{aaronson2008limits,
  title={The limits of quantum},
  author={Aaronson, Scott},
  journal={Scientific American},
  volume={298},
  number={3},
  pages={62--69},
  year={2008},
  publisher={JSTOR}
}

@book{stinson2018,
   title =     {Cryptography. Theory and Practice},
   author =    {Stinson, Douglas Robert and Paterson, Maura B.},
   publisher = {CRC Press},
   isbn =      {978-1-1381-9701-5},
   year =      {2018},
   series =    {},
   edition =   {4},
   volume =    {},
}

@book{knuthv2,
 author = {Knuth, Donald E.},
 title = {The Art of Computer Programming, volume 2, seminumerical algorithms},
 edition = {3},
 year = {1997},
 ISBN = {978-0-201-89684-8},
 publisher = {Addison-Wesley Longman Publishing Co., Inc.},
 address = {Boston, MA, USA},
}

@article{brent1980,
  author       = {Brent, Richard~P.},
  title        = {An Improved Monte Carlo Factorization Algorithm},
  journaltitle = {BIT 20: 176–184, doi:10.1007/BF01933190},
  date         = {1980}
}

@book{gathen,
   title =     {Modern Computer Algebra},
   author =    {Joachim von zur Gathen and Jürgen Gerhard},
   publisher = {Cambridge University Press},
   isbn =      {978-1-107-03903-2}, %% 1-107-03903-7
   year =      {2013},
   series =    {},
   edition =   {3},
   volume =    {},
}

@article{lawson2015,
  title={Odd orders in {Shor's} factoring algorithm},
  author={Thomas Lawson},
  journal={Quantum Information Processing},
  volume={14},
  number={3},
  pages={831--838},
  year={2015},
  publisher={Springer}
}

@book{hac,
   title =     {Handbook of Applied Cryptography},
   author =    {Menezes, A. J. and Oorschot, P. C. and Vanstone, S. A.},
   publisher = {CRC Press},
   isbn =      {0-8493-8523-7},
   year =      {1996},
   series =    {Discrete Mathematics and Its Applications},
   edition =   {3},
   volume =    {},
}

@inproceedings{shor94,
 author = {Peter W.~Shor},
 title = {Algorithms for Quantum Computation: Discrete Logarithms and Factoring},
 booktitle = {Proceedings of the 35th Annual Symposium on Foundations of Computer Science},
 series = {SFCS '94},
 year = {1994},
 isbn = {0-8186-6580-7},
 pages = {124--134},
 numpages = {11},
 doi = {10.1109/SFCS.1994.365700},
 acmid = {1399018},
 publisher = {IEEE Computer Society},
 address = {Washington, DC, USA},
}

@article{shor97,
   title={Polynomial-Time Algorithms for Prime Factorization and Discrete Logarithms on a Quantum Computer},
   volume={26},
   ISSN={1095-7111},
   url={http://dx.doi.org/10.1137/S0097539795293172},
   DOI={10.1137/s0097539795293172},
   number={5},
   journal={SIAM Journal on Computing},
   publisher={Society for Industrial & Applied Mathematics (SIAM)},
   author={Peter W.~Shor},
   year={1997},
   month={10},
   pages={1484–1509}
}

@article{shor99,
  title={Polynomial-Time Algorithms for Prime Factorization and Discrete Logarithms on a Quantum Computer},
  author={Peter W.~Shor},
  journal={SIAM Review},
  volume={41},
  number={2},
  pages={303--332},
  year={1999},
  publisher={SIAM}
}

@article{miller1976,
  title={Riemann's Hypothesis and Tests for Primality},
  author={Gary L.~Miller},
  journal={Journal of computer and system sciences},
  volume={13},
  number={3},
  pages={300--317},
  year={1976},
  publisher={Academic Press}
}

@article{leander2002,
  title={Improving the Success Probability for {Shor's} Factoring Algorithm},
  author={Leander, Gregor},
  journal={arXiv preprint quant-ph/0208183},
  year={2002}
}

@incollection{xu2018,
  title={Improving the Success Probability for {Shor's} Factorization Algorithm},
  author={Xu, Guoliang and Qiu, Daowen and Zou, Xiangfu and Gruska, Jozef},
  booktitle={Reversibility and Universality},
  pages={447--462},
  year={2018},
  publisher={Springer}
}

@misc{cao2005,
    author = {Zhengjun Cao},
    title = {A Note on Shor's Quantum  Algorithm for Prime Factorization},
    howpublished = {Cryptology ePrint Archive, Report 2005/051},
    %%year = {2005},
    url = {https://ia.cr/2005/051},
}

@inproceedings{pqrsa,
  title={Post-quantum {RSA}},
  author={Daniel J. Bernstein and Nadia Heninger and Paul Lou and Luke Valenta},
  booktitle={8th International Workshop, PQCrypto2017},
  pages={311--329},
  year={2017},
  organization={Springer},
  issn={1611-3349},
  isbn={978-3-319-59879-6},
}

@misc{johnston2017,
    author = {Anna M.~Johnston},
    title = {Shor's Algorithm and Factoring: Don't Throw Away the Odd Orders},
    howpublished = {Cryptology ePrint Archive, Report 2017/083},
    year = {2017},
    url = {https://ia.cr/2017/083},
}

@inproceedings{extension,
   author="Chicayban Bastos, Daniel
   and Kowada, Luis Antonio",
   title="How to detect whether Shor's Algorithm succeeds against large integers without a quantum computer",
   booktitle="XI Latin and American Algorithms, Graphs and Optimization Symposium",
   year="2021",
   pages="132--138",
}

@book{gauss,
   title =     {Disquisitiones Arithmeticae},
   author =    {Gauss, Carl F.},
   publisher = {Springer-Verlag},
   %% address   = {New York, Berlin, Heidelberg, Tokyo},
   isbn =      {0-387-96254-9},
   year =      {1986},
   series =    {},
   edition =   {English Edition},
   volume =    {},
}

@book{hardy1975intro,
  title={An introduction to the theory of numbers},
  author={Hardy, Godfrey Harold and Wright, Edward Maitland},
  year={1975},
  publisher={Oxford University Press},
  ISBN={0-19-853310-7} 
  % This ISBN should be 0-19-853310-1, but it's written 0-19-853310-7
  % on the very 4th edition.  It will probably be more useful to the
  % reader to have the wrong number than to have the correct one,
  % since these numbers are usually looked up in the databases (which
  % most likely have the wrong number) and by the reader himself.  
  % -- Daniel, February 27th 2020.
}

@book{qc,
   title =     {Quantum computation and quantum information},
   author =    {Michael A. Nielsen and Isaac L. Chuang},
   publisher = {Cambridge University Press},
   isbn =      {0-521-63503-9}, %% , 978-0-521-63503-5, 0-521-63235-8
   year =      {2004},
   series =    {Cambridge Series on Information and the Natural Sciences},
   edition =   {1},
   volume =    {},
}

@techreport{rabin1979,
  title={Digitalized signatures and public-key functions as intractable as factorization},
  author={Rabin, Michael O.},
  year={1979},
  number = {MIT/LCS/TR-212},
  month={1},
  institution = {Massachusetts Institute of Technology},
  address = {Cambridge, Massachusetts},
  url = {https://bit.ly/2WSjSXL},
  %%https://apps.dtic.mil/dtic/tr/fulltext/u2/a078415.pdf
}

@misc{grosshans2015,
    title={Factoring Safe Semiprimes with a Single Quantum Query},
    author={Frédéric Grosshans and Thomas Lawson and François Morain and Benjamin Smith},
    year={2015},
    eprint={1511.04385},
    archivePrefix={arXiv},
    primaryClass={quant-ph}
}

@article{rsa,
 author = {Rivest, Ronald L. and Shamir, Adi and Adleman, Leonard},
 title = {A Method for Obtaining Digital Signatures and Public-key Cryptosystems},
 journal = {Commun. ACM},
 issue_date = {Feb. 1978},
 volume = {21},
 number = {2},
 month = feb,
 year = {1978},
 issn = {0001-0782},
 pages = {120--126},
 numpages = {7},
 url = {https://goo.gl/J9uoRP},
 doi = {10.1145/359340.359342},
 acmid = {359342},
 publisher = {ACM},
 address = {New York, NY, USA},
}

@inproceedings{lenstra1994,
  title={Factoring},
  author={Lenstra, Arjen K.},
  booktitle={International Workshop on Distributed Algorithms},
  pages={28--38},
  year={1994},
  organization={Springer}
}

@book{joy,
  title={The Joy of Factoring},
  author={Wagstaff, Samuel~S.},
  volume={68},
  year={2013},
  publisher={American Mathematical Society},
  ISBN={978-1-4704-1048-3}
}

@article{pomerance1982analysis,
  title={Analysis and comparison of some integer factoring algorithms},
  author={Pomerance, Carl},
  journal={Computational methods in number theory},
  pages={89--139},
  year={1982},
  publisher={Math. Centrum}
}

@article{beckman1996,
  title = {Efficient networks for quantum factoring},
  author = {Beckman, David and Chari, Amalavoyal N. and Devabhaktuni, Srikrishna and Preskill, John},
  journal = {Phys. Rev. A},
  volume = {54},
  issue = {2},
  pages = {1034--1063},
  numpages = {0},
  year = {1996},
  month = {8},
  publisher = {American Physical Society},
  doi = {10.1103/PhysRevA.54.1034},
  %url = {https://link.aps.org/doi/10.1103/PhysRevA.54.1034}
}

@misc{beauregard2002,
    title={Circuit for Shor's algorithm using 2n+3 qubits},
    author={Stephane Beauregard},
    year={2002},
    eprint={quant-ph/0205095},
    archivePrefix={arXiv},
    primaryClass={quant-ph}
}

@book{hein1996,
   title =     {Discrete Mathematics},
   author =    {James L. Hein},
   publisher = {Jones and Bartlett Publishers},
   isbn =      {0-86720-496-6},
   year =      {1996},
   %%series =    {},
   %%edition =   {},
   %%volume =    {},
}

@article{aks2004,
  title={PRIMES is in P},
  author={Agrawal, Manindra and Kayal, Neeraj and Saxena, Nitin},
  journal={Annals of mathematics},
  pages={781--793},
  year={2004},
  publisher={JSTOR}
}

@Article{vedral,
  author =   {V. Vedral AND A. Barenco AND A. Ekert},
  title =    {Quantum networks for elementary arithmetic operations},
  journal =      {Physical Review A},
  year =     {1996},
  OPTkey =   {},
  volume =   {54},
  number =   {1},
  pages =    {147--153},
}

@Article{meter,
  author = 	 {R. V. Meter and K. M. Itoh},
  title = 	 {Fast quantum modular exponentiation},
  journal = 	 {Physical Review A},
  year = 	 {2005},
  OPTkey = 	 {},
  volume = 	 {71},
  number = 	 {5},
  pages = 	 {052320},
  OPTmonth = 	 {},
  OPTnote = 	 {},
  OPTannote = 	 {}
}

@article{shor1997polynomial,
  title={Polynomial-Time Algorithms for Prime Factorization and Discrete Logarithms on a Quantum Computer},
  author={Shor, Peter W.},
  journal={SIAM Journal on Computing},
  volume={26},
  number={5},
  pages={1484--1509},
  year={1997},
  publisher={SIAM}
}

@book{nielsen2004,
   title =     {Quantum computation and quantum information},
   author =    {Michael A. Nielsen and Isaac L. Chuang},
   publisher = {Cambridge University Press},
   isbn =      {0-521-63503-9},
   year =      {2004},
   series =    {Cambridge Series on Information and the Natural Sciences},
   edition =   {10th anniversary},
   volume =    {},
}
\end{document}